\documentclass[10pt,journal,compsoc]{IEEEtran}

\usepackage{graphicx}
\usepackage[caption=false,font=footnotesize]{subfig} 
\usepackage{amsmath,amssymb,amsthm}
\usepackage[boxruled,vlined,linesnumbered]{algorithm2e}
\usepackage{kbordermatrix}
\usepackage{booktabs}
\usepackage{threeparttable}
\usepackage{adjustbox}
\usepackage{xcolor} 

\newtheorem{lemma}{Lemma}
\newtheorem{theorem}{Theorem}

\newtheorem{example}{Example}

\renewcommand*{\algorithmcfname}{Procedure}
\begin{document}
\title{Computing All-Pairs Shortest Communication Time Paths for Periodically Varying Satellite Networks}

\author{Koushik Sinha \IEEEmembership{Senior Member,~IEEE}, Dibakar Saha, Sai Puppala and Bhabani P. Sinha \IEEEmembership{Fellow,~IEEE}
\IEEEcompsocitemizethanks{\IEEEcompsocthanksitem Koushik Sinha and Sai Puppala are with the School of Computing, Southern Illinois University, Carbondale, USA.
Dibakar Saha is with the Department of Computer Applications, National Institute of Technology Raipur, Chhattisgarh, India.
Bhabani P Sinha is with the Department of Computer Science and Engineering, Sister Nivedita University, Kolkata, India\protect\\
E-mail: {koushik.sinha@cs.siu.edu, dibakar.saha10@gmail.com, sai.puppala@siu.edu, sinha.bhabani@gmail.com}
}\thanks{}}

\IEEEcompsoctitleabstractindextext{
\begin{abstract}
We address the problem of designing an efficient algorithm for computing the all-pairs shortest-time paths in time-varying networks, specifically targeting the dynamic topological characteristics of low earth orbit (LEO) satellite constellations. We model such temporal graphs as $G = (V, E_t)$, where $V=\{1, 2,\ldots , n\}$ is the set of $n$ vertices in the graph, and $E_t$ is the set of $M$ time-varying edges considered over a total time window of size $T = M\tau$, $\tau$ being the duration of each interval in that time window. For such a network, our algorithm finds the all-pair shortest paths in $O(M^2n^3)$-time. A key advantage of our approach is its computational structure, which mirrors two matrix multiplication steps, first in the forward direction and then in the backward direction. This structure of the algorithm makes it exceptionally suitable for parallelization on SIMD architectures like GPUs, thereby offering a significant speedup over traditional shortest-path algorithms and enabling efficient, network-wide message routing in rapidly changing LEO satellite constellations.
\end{abstract}
}
\maketitle

\IEEEdisplaynotcompsoctitleabstractindextext
\IEEEpeerreviewmaketitle

\section{Introduction}

 \IEEEPARstart{W}e are witnessing an era of unprecedented expansion in satellite launches, with ambitious plans for multi-orbit, mega-constellations of satellites in the coming decade, involving thousands of interconnected satellites spanning low earth orbit (LEO), medium earth orbit (MEO), and geostationary orbit (GEO). The future of connectivity is quite literally ascending into orbit, with LEO satellite mega-constellations projected to play a pivotal role in that transformation as inter-satellite links (ISLs) become increasingly central to constellation-wide routing and capacity management \cite{wang2025ntn6g,cen2024satflow}. With a significantly lower cost to build and launch compared to MEO and GEO satellites, LEO satellites are an extremely attractive proposition for future networks, especially in the context of 6G non-terrestrial networks (NTNs), as their shorter distances from the surface of the Earth translate to lower latency and higher link bandwidth \cite{wang2025ntn6g}. As such, the concept of mega-constellations of LEO satellites for networking has gained significant traction in recent times, with major operators and programs advancing large-scale deployments and regulatory authorizations \cite{fccStarlinkGen2DA24-1193,telesatLightspeed198,onewebEutelsatLeMonde2025,chinaGuowangSpaceCom2025}.

 The potential of such space-based networks extends far beyond just providing networking services --- the ability to host in-orbit computing, thereby transitioning satellites into an \textsl{integrated communication and computational fabric} for offloading computation from terrestrial data centers \cite{yin2025orbitalEdgeComputingSurvey}. In May 2025, China reported the launch of an initial set of satellites toward a space computing constellation for collaborative AI-oriented workloads \cite{chinaThreeBodyComputingConstellation2025}. Related plans have also been discussed publicly for India and Europe in the context of space-based computing and satellite-network capability roadmaps \cite{isroSpaceDataCentersTheWeek2026,esaMilitaryAIEOConstellationFT2025}. Similarly, commercial providers have announced space and satellite cloud/edge offerings and supporting ground-to-space service infrastructure \cite{awsGroundStation,microsoftAzureSpaceOffPlanet2021,axiomAwsRelease2023}. Considering the vast computational resources circling in orbit across the earth, such an in-orbit distributed or edge computing platform can host a myriad of applications, including disaster response, situational awareness for space traffic management, and onboard AI-based filtering and prioritization for Earth-observation data \cite{yin2025orbitalEdgeComputingSurvey}.

 A key challenge in harnessing LEO satellite constellations for networking, or in-orbit computing, arises from the inherent temporal dynamics of LEO satellite networks. High-speed orbital motion induces time-varying ISL characteristics that lead to dynamic topology changes as satellites traverse deterministic orbital configurations. This predictability enables forward-looking planning and routing mechanisms that explicitly incorporate time dependence \cite{cen2024satflow,ron2025dteg}. By leveraging the predictable movement of satellites, the network topology for future time slots can be accurately forecasted, and such satellite networks can be represented naturally as time-evolving (temporal) graphs where edges appear and disappear over time \cite{longa2023tmlr,klobas2024sand}.

Such a temporal graph can be derived as an ordered list or a stream of timed events, such as additions or deletions of edges as connectivity becomes feasible or infeasible. Thus, the network's operational time can be represented as a temporal graph consisting of several static graphs, or \textit{snapshots}, that store the active edges for each time point in the lifetime of the graph. While a static graph representation for a highly dynamic network topology scenario as encountered in satellite networks suffers from loss of topology-change information, a temporal graph representation with appropriate time granularity is able to capture the full dynamics of the satellite constellation \cite{longa2023tmlr}. 

\subsection{Basic Idea}
We first illustrate the utility of temporal graphs for various application use-cases, followed by how our proposed \textsl{all-pairs shortest path algorithm} can benefit space-based networking and computing.


Furthermore, we often need to distinguish between the distance from a vertex $p$ to another vertex $q$ in terms of the number of traversed edges and the actual time taken to move from $p$ to $q$. This is because one may have to wait at some intermediate vertex for a long time to get the proper connecting edge of the path from $p$ to $q$. Accordingly, multiple temporal-path objectives are used in modern temporal-graph formulations, including shortest-duration (fastest), earliest-arrival, and shortest-length variants \cite{cauvi2025foremost,klobas2024sand,jain2024itg}.


As in-orbit networking and computing becomes a reality, central to the viability of such an infrastructure is efficient computation of up-to-date knowledge of the temporal graph topology characteristics. Accordingly, the \textsl{all-pairs shortest path algorithm} (APSP) is an attractive candidate that can pre-calculate the shortest path between all pairs of routers. Results of APSP computation can be used to populate routing tables. As noted earlier, a temporal graph induced by satellites in LEO exhibits periodic structure at the orbit timescale; for example, representative LEO orbits such as the ISS complete one revolution in roughly $\sim$90 minutes, illustrating the short recurrence times that motivate orbit-window precomputation. Thus, by executing a temporal all-pairs shortest path algorithm over snapshots captured over one (or multiple) orbital periods, one can pre-calculate temporal shortest paths between any pair of satellites. The resulting temporal shortest paths or routing tables can then be uploaded to the satellites from the ground stations so that whenever a communication request is generated, the appropriate temporal shortest path can be used for near-instantaneous path determination. Moreover, in the event a satellite fails or an ISL's properties change significantly, pre-computing all-pair temporal shortest paths for contingency scenarios can allow the network to quickly reroute traffic and avoid significant overhead.

\subsection{Our Contribution}

This paper addresses the computation of the shortest time paths between all pairs of vertices in the temporal graph representation of satellite networks. We present an elegant approach for this purpose which essentially involves the multiplication of specific matrices at the core of our proposed algorithm. The rest of the paper is organized as follows: Section \ref{sec:Rw} presents a brief literature survey. In Section \ref{sec:model}, we present the preliminaries and network model. Section \ref{sec:meth} presents the methodology. Section \ref{sec:algor} describes the proposed algorithm. We conclude with the areas of further exploration of this problem in Section \ref{sec:con}.

\section{Related Work} \label{sec:Rw}

Recent work has continued to expand the algorithmic toolkit for temporal-path problems beyond classical earliest-arrival formulations. Jin et al.~\cite{jin2023mcpe} study enumerating temporal $s$--$t$ paths under multiple constraints using temporal subgraph construction and pruning, highlighting the computational challenges created by time-order constraints and large search spaces. Jain and Sahni~\cite{jain2024itg} provide polynomial-time shortest-path algorithms for \emph{interval temporal graphs}, demonstrating that stronger temporal structure can yield tractable shortest-path computation beyond general contact-sequence models.  Parallel and hardware-accelerated temporal-path computation has also advanced. Mithinti et al.~\cite{mithinti2024gpu} propose GPU-based parallel algorithms for the fastest-path-duration problem in temporal graphs, indicating that matrix-/scan-like temporal primitives can be mapped effectively to massively parallel architectures. Complementing this, Yu et al.~\cite{yu2025fastgt} address the overhead of transforming temporal graphs into DAG-style representations for execution, proposing redundancy-aware GPU transformation that improves end-to-end performance of temporal analytics pipelines. These works collectively motivate designs that (i) preserve temporal feasibility exactly (e.g., run/interval constraints) and (ii) expose regular computation patterns suitable for acceleration. Beyond shortest-path computation, recent theoretical work has revisited \emph{periodicity} and feasibility constraints that arise naturally in recurring systems. Klobas et al.~\cite{klobas2024sand} initiate the temporal-graph realization problem with respect to fastest-path durations, including the periodic setting. Their results reinforce that periodic temporal structure is both practically relevant and algorithmically nontrivial, motivating methods that exploit periodicity without collapsing dynamics into static aggregation. In parallel, broad surveys such as Longa et al.~\cite{longa2023tmlr} summarize the state of temporal graph learning, emphasizing that temporal edges fundamentally change reachability and path semantics, and that realistic dynamic systems require time-aware modeling choices.

For LEO mega-constellations, multiple recent studies focus on coping with topology dynamics, laser-ISL operational constraints, and scalable planning. Bhattacharjee et al.~\cite{bhattacharjee2024taes} consider \emph{on-demand} routing when dynamic laser inter-satellite links incur non-negligible setup delays, presenting optimization and heuristic schemes evaluated on Starlink-like constellations. Ron et al.~\cite{ron2025dteg} propose a dynamic time-expanded-graph (DTEG) based approach for time-dependent topology optimization under multi-objective criteria (capacity, latency, churn), reflecting growing interest in explicitly time-dependent design objectives.  At a higher planning layer, Cen et al.~\cite{cen2024satflow} introduce SatFlow, a scalable framework that jointly plans topology, traffic allocation, and terminal power allocation for mega-constellations, addressing both dynamics and operational cost at scale. Demand-aware topology design has also been explored in very recent work: Dehshali et al.~\cite{dehshali2026starfield} propose a traffic-driven heuristic (Starfield) to orient ISLs toward non-uniform demand patterns, reporting improvements over baseline topologies in packet-level simulation.

Complementing routing and planning, link-configuration and feasibility modeling continues to evolve. Ren et al.~\cite{ren2025islalloc} study ISL allocation as a multi-objective optimization problem under dynamic topology and resource limits. At the physical-layer/network interface, Gioulis et al.~\cite{gioulis2025photonics} present a constellation-aware optical ISL model that evaluates power requirements across an orbital period, supporting more realistic feasibility constraints for temporal adjacency. Finally, learning-based routing in dynamic LEO settings has also gained traction; Riyad et al.~\cite{riyad2025qlearning} report Q-learning based routing over snapshot-derived topologies, reflecting a trend toward adaptive routing under time-varying connectivity.

Taken together, recent literature shows (i) renewed algorithmic attention to temporal shortest/fastest path computation and its acceleration, and (ii) rapidly growing interest in time-dependent routing/planning for LEO mega-constellations with laser ISLs. Our work targets the intersection of these themes by focusing on all-pairs shortest \emph{communication-time} paths under discrete-time feasibility constraints, with an implementation structure intentionally aligned with parallelizable primitives.

\section{Network model and Assumptions}
\label{sec:model}
This section introduces the basic preliminaries of the network model and the assumptions.
\subsection{Preliminaries}
We introduce a model of the temporal graph to be used for our later discussions, which is formulated around the concepts of time-varying connectivity in a temporal graph as given in \cite{nicosia2013}, but with some modifications. The following components describe our model:

\subsubsection{Temporal graph representation}
We denote a temporal graph as $G = (V,E_t)$, where $V  = \{ 1, 2, \cdots, n\}$ is the set of $n$ vertices in the graph and $E_t$ is the set of time-varying edges considered over a time window of size $T$ starting from time $t = 0$ to $t = T$. This time window of observation will be represented as $W = [0,T]$.

\subsubsection{Time Window}
In general, a temporal edge $e_{pq} \in E_t$  connecting two vertices $p$ and $q$ during this time window $W$ will be represented by a tuple $(p, q, [t_1, \Delta t_1], [t_2, \Delta t_2], \cdots, [t_k, \Delta t_k])$, where the vertices $p$ and $q$ are connected only during the non-overlapping time durations $[t_i, t_i+ \Delta t_i], i = 1, 2,\cdots, k$ with the time ordering of $0 \leq t_1 < t_1 + \Delta t_1 < t_2 <  t_2 + \Delta t_2 < \cdots < t_k < t_k + \Delta t_k \leq T$. An example of such time-varying connections is shown by the horizontal line segments in Fig.~\ref{connection-timings} for $k = 3$. We refer to these time durations $[t_i, t_i+ \Delta t_i], i = 1, 2,\cdots, k$ as the {\em connectivity durations} for the vertices $p$ and $q$ during the observation window $[0,T]$, and denote these by $CD_{pq}[0,T]$ or simply $CD_{pq}$ when the time window $[0,T]$ is implied from the context.
\begin{figure}
	\centering
	\includegraphics[width=\linewidth]{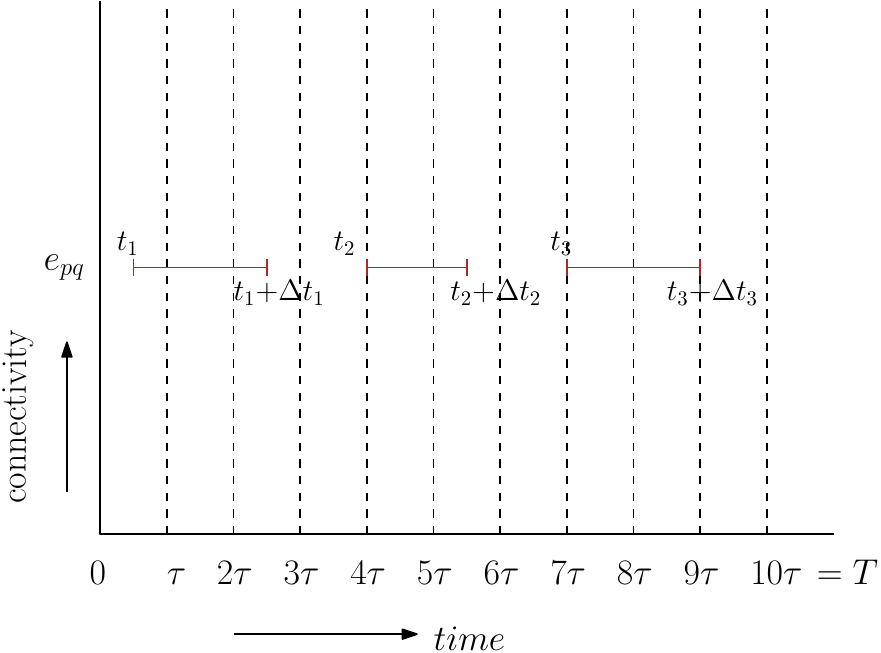}
	\caption{Connectivity durations of edge $e_{pq}$}
	\label{connection-timings}
\end{figure}

\subsubsection{Time slots for each snapshot}
 We divide the time window $[0,T]$ in equal intervals of time duration $\tau$ (refer to Fig.~\ref{connection-timings} for a diagrammatic representation), where a snapshot of the connectivity is taken during each such interval. Each interval's length $\tau$ may be set to the most petite possible duration used to observe the graph's time-varying nature with maximum temporal resolution. The exact value of $\tau$ may be chosen depending on the requirements of typical real-life applications, e.g.,  phone calls, email communications, etc., for which the temporal graph model will be used.

 \begin{example}
 For example, the message transmission delay can be computed as $\cal{D}=\frac{Q_m}{B}$, where $Q$ is the size of the message $m$ and $B$ is the bit rate of a channel in a particular band \cite{shah2019x}.
 Therefore, consider a node $v_i$ starts the
transmission of message $m$ of size $10MB$ at a time slot $t$ over GSM band $20MHz$. Then, the communication latency can be computed as:  ${\cal{D}}=\frac{10}{20}=0.5$ seconds. But in case the same message will be sent through GSM Band $10MHz$, then ${\cal{D}}$ will be ${\cal{D}}=\frac{10}{20}=1$ seconds.
Hence, we consider the value of $\tau$ depending on the current availability of the channel bandwidth and the size of the message to be sent. In this example, $\tau$ will be set to one second, and for each second, we take the snapshot accordingly.
\label{ex:td}
\end{example}

 We assume that  $T = M \tau$, so that there will be a total of $M$ snapshots, numbered as $0, 1, 2, \cdots, M-1$ taken during these $M$ intervals $[0, \tau], [\tau, 2 \tau], \cdots, ... [(M-1) \tau, T]$, respectively to provide the total connectivity information of every edge of the temporal graph during the time window $[0,T]$ under consideration. An interval $I = j$ will refer to the time duration from $j\tau$ to $(j+1)\tau$.

 Two vertices $p$ and $q$ will be connected by an edge $e_{pq}$ during the $j^{th}$ snapshot, $j \in \{0,1,2, \cdots, M-1\}$, if there is a continuous connectivity of duration at least $\lambda \tau$ $( 0< \lambda < 1)$ during  the interval $[j \tau, (j+1) \tau]$ which also has a non-null intersection with at least one of the $k$ connectivity durations  $[t_i, t_i+ \Delta t_i], i = 1, 2,\cdots, k$ of the vertices $p$ and $q$. Note that, $\lambda$ can be a priori fixed depending on the specific application requirements.

\subsubsection{Connectivity String}
The connectivity information during the time window $[0,T]$ from the vertex $p$ to the vertex  $q$ of the temporal graph will then be represented by a binary string of length $M$, termed as the {\em connectivity string}, in which each bit corresponds to a snapshot interval so that a bit will be set to 1 if the connectivity between $p$ and $q$ was there during the corresponding time interval, and set to 0, otherwise.

Note that in each snapshot interval, if $\tau$ is greater or equal to a certain value, then only we consider the connection set the value 1. For example, considering the Example \ref{ex:td}, if $\tau\geq 1$, then we set the bit value 1 other 0.

This connectivity string from $p$ to $q$ during the time window $[0,T]$ will be denoted by $CS_{pq}[0,T]$ or simply $CS_{pq}$ when the time window $[0,T]$ is implied from the context. We assume that every edge of the temporal graph is symmetric (undirected) and hence, $CS_{pq}[0,T] = CS_{qp}[0,T]$.

\begin{example} \label{example1}
Referring to the timing diagram  for connection as given in Fig.~\ref{connection-timings} with $M = 10$, the connectivity between the vertices $p$ and $q$ is expressed by the 10-bit binary string 1110110110 (the most significant bit corresponds to the snapshot interval $[0, \tau]$ ).
\end{example}

\subsection{Assumtions}
\begin{itemize}
    \item  We assume that a finite amount of time is needed to traverse every edge in the path, and the snapshot interval $\tau$ is so small that this time to traverse an edge is always more than $\tau$. Without loss of generality, we assume for simplicity that this time to traverse an edge can be expressed as some non-zero integral multiple of $\tau$.

\item During the whole time window $[0,T]$, the distance $d_{pq}$ from $p$ to $q$, as well as the time $t_{pq}$ to traverse the edge $e_{pq}$ from  $p$ to $q$ (when it exists throughout the period $t_{pq}$) will be assumed to be constant. Given these information, it may be noted that when we start traversing the edge $e_{pq}$ at any time during the interval $I = j$, i.e., during $[j\tau, (j+1)\tau]$, this connecting edge $e_{pq}$  must exist continuously during the interval $[j\tau, (j + t_{pq})\tau]$  for the completion of the traversal.

\item A temporal graph will be represented by an $n \times n$ adjacency matrix $A$ whose $(p,q)^{th}$ entry will be a 3-tuple given as $A(p,q) = (CS_{pq}, d_{pq}, t_{pq})$.
\end{itemize}

\section{Methodology}
\label{sec:meth}
Consider a temporal path $p_1p_2 \cdots p_m$ from a vertex $p_1$ to another vertex $p_m$ that exists in the temporal graph $G = (V, E_t)$ during the time window $W = [0,T]$ under consideration.
After reaching a vertex $p_j, j = 2, 3, \cdots, m-1$ along the path from the vertex $p_{j-1}$, a connection from $p_j$ to $p_{j+1}$ following the trail should be continuously available (without any break) in any subsequent $t_{p_j p_{j+1}}$ snapshot intervals. To satisfy this requirement, one may have to wait at the vertex $p_j$ for some time until the required connectivity to the vertex $p_{j+1}$ is guaranteed. We illustrate this situation by the following example.

\begin{example}\label{example2}
Consider a temporal graph with five vertices for which the vertex set is $V = \{1,2, 3, 4, 5\}$. The connecting edges are time-varying in nature with the aggregated static graph as shown in Fig.~\ref{example-static-graph} during an observation window $[0,10\tau]$. The time varying nature of the connecting edges $e_{12}, e_{23}, e_{14}, e_{43}, e_{25}, e_{35}$ and $e_{45}$ are given by the following connectivity strings, each of length 10:
\begin{figure}
	\centering
	\includegraphics[width=0.8\linewidth]{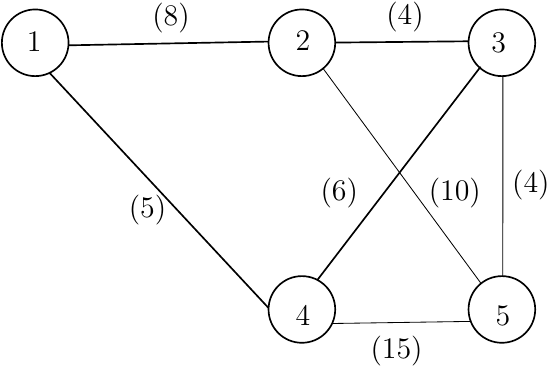}
	\caption{Aggregated static graph for the example situation}
	\label{example-static-graph}
\end{figure}
\begin{align*}
Time~intervals~ (I)~ :~  &0  &1  &2  &3  &4  &5  &6  &7  &8  &9 \\
CS_{12}~ =~           &0  &0  &0  &1  &1  &1  &0  &1  &1  &1 \\
CS_{23}~ = ~          &0  &1  &1  &1  &0  &1  &1  &0  &0  &0 \\
CS_{14}~ =~           &1  &1  &0  &1  &1  &1  &0  &1  &1  &0 \\
CS_{43}~ = ~          &0  &0  &0  &0  &1  &1  &0  &1  &1  &0 \\
CS_{25}~ = ~          &0  &0  &0  &1  &1  &1  &1  &1  &0  &0 \\
CS_{35}~ = ~          &1  &1  &1  &0  &0  &0  &0  &1  &1  &0 \\
CS_{45}~ = ~          &0  &0  &0  &0  &0  &0  &0  &1  &1  &1 \\
\end{align*}

Let two time units (each unit is equal to the length of an interval) be needed to traverse any of these four edges, and we need to reach the vertex 3 from vertex 1.  One possibility is that we start at interval $I = 0$ to reach vertex 4 traversing the edge $e_{14}$ during the intervals 0 and 1, wait there for two time intervals and  then traverse the edge $e_{43}$ during the intervals 4 and 5. Thus, we reach vertex 3 at interval $I = 5$ and the total time of travel is 6 units.

Another possibility is that we start from vertex 1 at interval $I = 3$, reach vertex 2 at interval $I = 4$, and then reach vertex 3 at $I = 6$. In this case, the total time of travel is only 4 units although the arrival time is at interval $I = 6$.

Yet another possibility exists as follows : start from vertex 1 at interval $I = 4$ to reach vertex 4 at interval $I = 5$, wait at vertex 4 for one interval and then reach vertex 3 at $I = 8$, with a total travel time of 5 units.

The first of these three schemes is an example for reaching vertex 3 from vertex 1 by following  the earliest arrival time path, the second one is by following the fastest path or shortest time path, and the third one is by following the latest departure time path.

If the distances between the vertices  are given as $d_{12} = 8, d_{23} = 4, d_{14} = 5$, and $d_{43} = 6$ (shown in Fig.~\ref{example-static-graph} as the labels of the corresponding edges), then the path 1-4-3 is the shortest distance path from vertex 1 to vertex 3.
\end{example}
While traversing a path from $p$ to $q$, we refer to the time of start from the initial vertex $p$ of this path by the starting interval $I^{pq}_s$ and the time of arrival at the final vertex $q$ of the path by the finishing interval $I^{pq}_f$. We represent this path by a 3-tuple $(I^{pq}_s, I^{pq}_f, pred^{pq})$, where $pred^{pq}$ is the predecessor vertex of $q$ in the path. $pred^{pq}$  will enable us to trace this path from $p$ to $q$ backward, starting from $q$. It may be noted that we would like to explore all possible simple paths only (without containing any loops) between every pair of vertices. However, there may be more than one such possible paths with different values of $pred^{pq}$. Further, it may also be possible that  with the same $pred^{pq}$, there may be different values of $I^{pq}_s$ and/or $I^{pq}_f$ defining a path from $p$ to $q$. We construct a path matrix $P$ whose $(p,q)^{th}$ entry will contain information about all such possible contact periods and will be represented as $(\nu_{pq}, (I^{pq}_s[1], I^{pq}_f[1], pred^{pq}[1])$, $(I^{pq}_s[2], I^{pq}_f[2], pred^{pq}[2])$, $\cdots$, $(I^{pq}_s[\nu_{pq}], I^{pq}_f[\nu_{pq}], pred^{pq}[\nu_{pq}]))$, where $\nu_{pq}$ is the number of possible paths from $p$ to $q$, and the 3-tuple $(I^{pq}_s[i], I^{pq}_f[i], pred^{pq}[i])$ contains information about the starting interval, finishing interval and the corresponding predecessor vertex of $q$ in the $i^{th}$ path, $1 \leq i \leq \nu_{pq}$.
As special cases, the entries of $P(p,p), p \in V$ will be set to $(0,(0,0,0))$ and also when there will be no path from $p$ to $q$, the entries of $P(p,q)$ will be set to $(0,(0,0,\infty))$.

With the above idea, the path matrix $P$ is first initialized based on 1-adjacency of the vertices, following the algorithm $initialize\_path\_matrix$ as shown in Procedure \ref{proc:init}. After that, to find the all-pairs shortest time paths, we would augment each entry of the path matrix with  information about other possible contacts (as expressed through a 3-tuple $(I^{pq}_s[i], I^{pq}_f[i], pred^{pq}[i])$), if any, through more than one edge. Thus, the $(p,q)^{th}$ entry of $P$ will finally contain information about all possible paths from $p$ to $q$ indicated by the $\nu_{pq}$ 3-tuples $(I^{pq}_s[i], I^{pq}_f[i], pred^{pq}[i])$, $1 \leq i \leq \nu_{pq}$. Finally, from all these $\nu_{pq}$ tuples, the one with the minimum value of $(I^{pq}_f[i] - I^{pq}_s[i])$ will be chosen as the path with the shortest time from $p$ to $q$.

To achieve the above objective of computing the final path matrix, we proceed as follows. Assume that we have already computed the contact information for a path from $p$ to $q$ with the corresponding entry in the path matrix $P(p,q)$ as  $(I^{pq}_s[i], I^{pq}_f[i], pred^{pq}[i])$. Suppose, from the vertex $q$, we would now like to reach another vertex $r$ adjacent to $q$ where $A(q,r) = (CS_{qr}, d_{qr}, t_{qr})$. It will be possible to reach the vertex $r$ from $q$ if we find some $j > I^{pq}_f[i]$ such that there is a run of $1's$ of length $t_{qr}$ in $CS_{qr}$ starting from the interval $j$, i.e., $CS_{qr}[j] = CS_{qr}[j+1] = \cdots = CS_{qr}[j+t_{qr}-1] = 1$, for some $j > I^{pq}_f[i]$. Once such a run is found,  the total time of traversal from $p$ to $r$ will be $j+t_{qr} - I^{pq}_s[i]$.
  The start time and finish time from $p$ to $r$ through $q$ will also be  accordingly evaluated as  $I^{pq}_s[i]$ and $j+t_{qr}-1$, respectively and $pred^{pr}[i]$ will be set to $q$.  The corresponding 3-tuple $(I^{pq}_s[i],j+t_{qr}-1,q)$ will then be inserted as a new 3-tuple in $P(p,r)$ in the sorted order of $(I^{pr}_s,I^{pr}_f)$ values provided there is no other existing 3-tuple with the same values of $(I^{pr}_s,I^{pr}_f)$, to avoid duplicate entries which might be generated in the above process of finding the possible contacts.

The above operation will be repeated with all the $\nu_{pq}$ tuples  $(I^{pq}_s[i], I^{pq}_f[i], pred^{pq}[i])$, $1 \leq i \leq \nu_{pq}$.
 Collectively, we denote all these operations symbolically by $P(p,q) \otimes A(q,r)$.  The algorithm for computing $P(p,q) \otimes A(q,r)$ is presented in the form of  $update\_all\_possible\_contacts(A,P,p,r)$ as given in  Procedure \ref{proc:apc}.

\renewcommand*{\algorithmcfname}{Procedure}
\setcounter{algocf}{0}
\begin{algorithm}[!h]
\scriptsize
 \KwIn{ $A,P,q,r,start,finish,flag$}
 \KwOut{Updated $start, finish, flag$}
 $count = 0$\;
    $flag= 0$\;
    $ j = start$\;
    \tcp{next run of length $t_{qr}$ or more}
    \While{$j <= M $ AND $ flag== 0$}
    {
        \eIf{$CS_{qr}[j] == 1$}
        {
            $count++$\;
            \If{$count == t_{qr}$ \tcp{a run of length $t_{qr}$ just found}}
            {
                $start = j - t_{qr} + 1$\;
            }
        }
        {
            \If{$count \geq t_{qr}$ \tcp{a run of length $t_{qr}$ or greater is found}}
            {
                $flag= 1$\;
                $finish = j-1$\;
                $break$\;
            }
            $count = 0$\;
        }
        $j++$\;
    }
\caption{\small $find\_next\_run(A,P,q,r,start,finish,flag)$}
\label{proc:forward}
\end{algorithm}

\begin{algorithm}[!ht]
\scriptsize
 \KwIn{ Adjacency matrix $A$}
 \KwOut{Path matrix $P$}
 \For{$p = 1$ to $n$}
    {
       \For{$r = 1$ to $n$}
        {
        	    $\nu_{pr} = 0$\;
                $I^{pr}_s[1] = 0$\;
                 $I^{pr}_f[1] = 0$\;
        		\eIf{$p == r$\tcp{diagonal entry}}
        		{
        		
        		     $pred^{pr}[1] = 0$\;
                }
                {
                 \tcp{for non diagonal entry}
                  $pred^{pr}[1] = \infty$\;
	               $start= 0$\;
                $finish=0$\;
                $k = 0$\;
                \While{$k<M-t_{pr}$}
                {
                    $find\_next\_run(A,P,p,r,start,finish,flag)$\tcp*{Execute Procedure \ref{proc:forward}}
                    \If {$flag== 1$}
                    {
                        $temp = start+ t_{pr} - 1$\;
                        \While{ $temp \leq finish$}
                        {
                           $k++$\;
                           $I^{pr}_s[k] = start$\;
                           $I^{pr}_f[k] = temp$\;
                           $start++$\;
                           $temp++$\;
                        }
                        $\nu_{pr} = k$\;
                        \For{$i = 1$ to $\nu_{pr}$}
                        {
                            $pred^{pr}[i] = p$\;
                        }
                    }
                }

            }
        }
    }
\caption{\small $initialize\_path\_matrix(A,P)$}
\label{proc:init}
\end{algorithm}

\begin{algorithm}[!h]
\scriptsize
 \KwIn{ $A,P,p,r$}
 \KwOut{Updated $P$}

 \For{$q = 1$ to $n$}
    {
    	
        \If{$q \neq r$}
        {
            \If{$\nu_{pq} > 0$ AND $\nu_{qr}> 0$}
            {
                $k = 1$\tcp*{$k$ indicates the position of the interval in $P(q,r)$ from left to right direction}
                $m=\nu_{pr}$\;
                \While{$k \leq \nu_{qr}$}
                {
                  \tcp{Looking for earliest connection interval in $P(q,r)$}
                    $flag_1= 0$\;
                    $j= \nu_{pq}$\tcp*{$j$ indicates the position of the interval in $P(p,q)$ traversed from right to left in the list}
                    \While{$j\geq 1$ AND $flag_1==0$}
                    {
                        \If{$ I^{qr}_s[k] > I^{pq}_f[j]$ }
                        {
                            $flag_2=0$\;

                            \For{ $i=1$ to $m$}
                            {
                            \tcp{avoid duplicate entries and entries with larger duration for a given finish time interval}
 								\eIf{$I^{pr}_f[i]==I^{qr}_f[k]$}
                                {
                                	\If{$I^{pr}_s[i]<I^{pq}_s[j]$}
                                	{
                                		$I^{pr}_s[i]=I^{pq}_s[j]$\;
                                		$pred^{pr}[i]=q$\;
                                	}
                                		$flag_2=1$\;
                                	
                                }
                                {
                                    $flag_2=1$\;
                                 }
                            }
							 \If{$flag_2==0$}
                               {
                                $m++$\;
                                $I^{pr}_s[m] = I^{pq}_s[j]$\;
                                $I^{pr}_f[m] = I^{qr}_f[k]$\;
                                $pred^{pr}[m] = q$\;
                                $flag_1= 1$\;
                                Insert this entry $(I^{pr}_s[m],I^{pr}_f[m],pred^{pr}[m])$ in $P(p,r)$ in sorted order with respect to the finish time $I^{pr}_f[m]$;
                            }

                       }
                        $j--$\;
                     }
                    $k++$\;
                }

                   $ \nu_{pr} = m$\;

            }

        }
    }
\caption{$update\_all\_possible\_contacts(A,P,p,r)$}
\label{proc:apc}
\end{algorithm}

With these preliminary ideas, we describe below an algorithm for computing all-pairs shortest time paths as described in the next section.

\section{Algorithm for Computing all-pairs Shortest Time Paths}
\label{sec:algor}

All-pairs shortest paths in static graphs can be very conveniently computed in two steps of a modified matrix multiplication scheme as proposed by Hu \cite{hu1967}. Let $D$ be the $n \times n$ distance matrix of the static graph in which the $(p,q)^{th}$ element $d_{pq}$ is set to the distance from the vertex $p$ to vertex $q$. The concept given in \cite{hu1967} is to compute the product (in some sense) of the matrix $D$ by itself ($D \times D$)  through in-place computation following some well-defined rules - i) an entry of the product matrix once evaluated will immediately replace the corresponding entry of only the left $D$ matrix in the above multiplication process, ii) the entries of the product matrix are computed once through a forward multiplication process and then through a backward multiplication process. In the forward multiplication process, the entries of the product matrix are evaluated in row-major order from row 1 to row $n$, in each row from left to right. In the backward multiplication process, entries are evaluated from row $n$ to row 1, in each row from right to left.

To compute the all-pairs shortest time paths in a temporal graph, we adopt the basic concept of in-place ordered matrix multiplication in Hu's algorithm mentioned above and introduce appropriate modifications on it to deal with the time-varying nature of the connecting edges. First, the matrix multiplication is equivalently replaced by an operation where, in computing an element of the product matrix, the multiplication of two elements of a matrix is replaced by $\otimes$ operation between $P$ and $A$ and second, the summation of these computed products is replaced by augmenting each of the $(p,q)^{th}$ entries, $1 \leq p, q \leq n$ of the product matrix with the newly generated 3-tuples containing information of additional contacts from $p$ to $q$ during the period under consideration. Finally, from each such $(p,q)^{th}$ entry of the product matrix, the 3-tuple with the minimum time between the start and finish intervals is chosen to find the shortest time path from $p$ to $q$. Our proposed algorithm for computing all-pairs shortest time paths of a temporal graph is shown in Algorithm \ref{all-pairs-st}, where we essentially perform the in-place computation of the path matrix $P(p,r), 1 \leq p,r\leq n$ through a forward and backward process of this equivalent matrix multiplication operation.

In this algorithm, we first initialize the path matrix $P(p,r), 1 \leq p,r \leq n$ by executing Procedure \ref{proc:init}. In this procedure, for each diagonal entry ($p=r$), we initialize $\nu_{pr}=I^{pr}_s[1]=I^{pr}_f[1]=pred^{pr}[1] = 0$. For non diagonal entry ($p\neq r$), we call $find\_next\_run$ (Procedure \ref{proc:forward}) to find all the consecutive runs of length $t_{pr}$ from $CS_{pr}\in A$, and set $\nu_{pr}=k$, where $k$ is the total number of runs of length $t_{pr}$, $1\leq k\leq M-t_{pr}$, and update all the entries $I^{pr}_s[k],I^{pr}_f[k],$ and $pred^{pr}[k]$, accordingly. Once we initialize the path matrix $P$, we start the forward and backward process by  executing $update\_all\_possible\_contacts$ (Procedure \ref{proc:apc}). In this procedure, for each entry of $P(p, r)$, we find all possible paths through the vertex $q$, $1\leq q\leq n$ and $q\neq r$, avoiding any duplicate entry ans also entries with larger duration for a given finish time interval $I^{pr}_f$. After computing all the intervals of $P(p,r)$, we call Procedure \ref{proc:min} to find the minimum time path among all the intervals in $P(p,r)$. Finally, we set the interval number  $\nu_{pr} = 1$ and update $I^{pr}_s[1], I^{pr}_f[1], pred^{pr}[1]$ by the minimum time path accordingly.
 Thus, Algorithm \ref{all-pairs-st} computes all-pairs shortest time communication paths in 6G networks based on temporal graph representation.

\begin{algorithm}[!h]
\scriptsize
 \KwIn{ $P$}
 \KwOut{Updated $P$}
 \For{$p = 1$ to $n$}
    {
       \For{$r = 1$ to $n$}
        {
            \If{$\nu_{pr} > 0$}
            {
                $min\_time = I^{pr}_f[1] - I^{pr}_s[1] + 1$\;
                $temp\_start = I^{pr}_s[1]$\;
                $temp\_finish = I^{pr}_f[1]$\;
                $temp\_predecessor = pred^{pr}[1]$\;
                $previous\_node = pred^{pr}[1]$\;
                \For{$j = 2$ to $\nu_{pr}$}
                {
                       $temp\_duration = I^{pr}_f[j] - I^{pr}_s[j] + 1$\;
                        \If{$temp\_duration < min\_time$}
                        {
                            $min\_time = temp_duration$\;
                            $temp\_start = I^{pr}_s[j]$\;
                            $temp\_finish = I^{pr}_f[j]$\;
                            $temp\_predecessor = pred^{pr}[j]$\;
                        }
                }
                $\nu_{pr} = 1$\;
                $I^{pr}_s[1] = temp\_start$\;
                $I^{pr}_f[1] = temp\_finish$\;
                $pred^{pr}[1] = temp\_predecessor$\;
            }
        }
    }
\caption{\small $find\_minimum\_time\_path(P)$}
\label{proc:min}
\end{algorithm}

\renewcommand*{\algorithmcfname}{Algorithm}
\begin{algorithm}[!ht]
\setcounter{algocf}{0}
\scriptsize
 \KwIn{ Adjacency matrix $A$}
 \KwOut{Path matrix $P$}
 \tcp{\bf Initialize path matrix with 1-adjacent information:}
 $initialize\_path\_matrix(A,P)$ \tcp*{Execute Procedure \ref{proc:init}}
\tcp{\bf Forward:}
 \For{$p = 1$ to $n$}
    {
       \For{$r = 1$ to $n$}
        {
            \If{$p \neq r$}
            {
           $update\_all\_possible\_contacts(A,P,p,r)$\tcp*{Execute Procedure \ref{proc:apc}}
            }
        }
    }
\tcp{\bf Backward:}
 \For{$p = n$ to $1$}
    {
       \For{$r = n$ to $1$}
        {
            \If{$p \neq r$}
            {
           $update\_all\_possible\_contacts(A,P,p,r)$\tcp*{Execute Procedure \ref{proc:apc}}
            }
        }
    }
\tcp{Compute shortest time path for each $P(p,q)$ from all its contacts}
$find\_minimum\_time\_path(P)$\tcp*{Execute Procedure \ref{proc:min}}
\caption{\small $all\_pair\_shortest\_time\_paths(A,P)$}
\label{all-pairs-st}
\end{algorithm}

\begin{example}\label{example3}

Consider again the graph used in Example \ref{example2} with the same connectivity strings. While executing Algorithm~\ref{all-pairs-st}, the path matrix is initialized as follows:

{
\scriptsize
\[
 P=\kbordermatrix{
    & 1 & 2 & 3 & 4 & 5 \\
  1  \hspace{-0.15cm} & \hspace{-0.15cm} \big(0,(0,0,0)\big)\hspace{-0.15cm} & \hspace{-0.15cm} \big(0,(0,0,\infty)\big)\hspace{-0.15cm} & \hspace{-0.15cm} \big(0,(0,0,\infty)\big)\hspace{-0.15cm} & \hspace{-0.15cm} \big(0,(0,0,\infty)\big)\hspace{-0.15cm} & \hspace{-0.15cm} \big(0,(0,0,\infty)\big)\hspace{-0.15cm} \\
  2  \hspace{-0.15cm} & \hspace{-0.15cm}  \big(0,(0,0,\infty)\big)\hspace{-0.15cm} & \hspace{-0.15cm} \big(0,(0,0,0)\big)\hspace{-0.15cm} & \hspace{-0.15cm} \big(0,(0,0,\infty)\big)\hspace{-0.15cm} & \hspace{-0.15cm} \big(0,(0,0,\infty)\big)\hspace{-0.15cm} & \hspace{-0.15cm} \big(0,(0,0,\infty)\big)\hspace{-0.15cm} \\
  3  \hspace{-0.15cm} & \hspace{-0.15cm} \big(0,(0,0,\infty)\big)\hspace{-0.15cm} & \hspace{-0.15cm} \big(0,(0,0,\infty)\big)\hspace{-0.15cm} & \hspace{-0.15cm} \big(0,(0,0,0)\big)\hspace{-0.15cm} & \hspace{-0.15cm} \big(0,(0,0,\infty)\big)\hspace{-0.15cm} & \hspace{-0.15cm} \big(0,(0,0,\infty)\big) \hspace{-0.15cm}\\
  4  \hspace{-0.15cm} & \hspace{-0.15cm}  \big(0,(0,0,\infty)\big)\hspace{-0.15cm} & \hspace{-0.15cm} \big(0,(0,0,\infty)\big)\hspace{-0.15cm} & \hspace{-0.15cm} \big(0,(0,0,\infty)\big)\hspace{-0.15cm} & \hspace{-0.15cm} \big(0,(0,0,0)\big)\hspace{-0.15cm} & \hspace{-0.15cm} \big(0,(0,0,\infty)\big)\hspace{-0.15cm} \\
  5  \hspace{-0.15cm} & \hspace{-0.15cm}\big(0,(0,0,\infty)\big)\hspace{-0.15cm} & \hspace{-0.15cm} \big(0,(0,0,\infty)\big)\hspace{-0.15cm} & \hspace{-0.15cm} \big(0,(0,0,\infty)\big)\hspace{-0.15cm} & \hspace{-0.15cm} \big(0,(0,0,\infty)\big)\hspace{-0.15cm} & \hspace{-0.15cm} \big(0,(0,0,0)\big)\\ [1ex]
}
\]

}

\noindent After initialization, the entries of the path matrix $P(p,q)$ are as follows:\\
{
        $~~~~P(1,1) =  \big( 0, (0, 0, 0)\big)$\\
        $~~~~P(1,2) =  \big( 4, (3, 4, 1) (4, 5, 1) (7, 8, 1) (8, 9, 1)\big)$\\
        $~~~~P(1,3) =  \big( 0, (0, 0, \infty)\big)$\\
        $~~~~P(1,4) =  \big( 4, (0, 1, 1) (3, 4, 1) (4, 5, 1) (7, 8, 1)\big)$\\
        $~~~~P(1,5) =  \big( 0, (0, 0, \infty)\big)$\\
        $~~~~P(2,1) =  \big( 4, (3, 4, 2) (4, 5, 2) (7, 8, 2) (8, 9, 2)\big)$\\
        $~~~~P(2,2) =  \big( 0, (0, 0, 0)\big)$\\
        $~~~~P(2,3) =  \big( 3, (1, 2, 2) (2, 3, 2) (5, 6, 2)\big)$\\
        $~~~~P(2,4) =  \big( 0, (0, 0, \infty)\big)$\\
        $~~~~P(2,5) =  \big( 4, (3, 4, 2) (4, 5, 2) (5, 6, 2) (6, 7, 2)\big)$\\
        $~~~~P(3,1) =  \big( 0, (0, 0, \infty)\big)$\\
        $~~~~P(3,2) =  \big( 3, (1, 2, 3) (2, 3, 3) (5, 6, 3)\big)$\\
        $~~~~P(3,3) =  \big( 0, (0, 0, 0)\big)$\\
        $~~~~P(3,4) =  \big( 2, (4, 5, 3) (7, 8, 3)\big)$\\
        $~~~~P(3,5) =  \big( 3, (0, 1, 3) (1, 2, 3) (7, 8, 3)\big)$\\
        $~~~~P(4,1) =  \big( 4, (0, 1, 4) (3, 4, 4) (4, 5, 4) (7, 8, 4)\big)$\\
        $~~~~P(4,2) =  \big( 0, (0, 0, \infty)\big)$\\
        $~~~~P(4,3) =  \big( 2, (4, 5, 4) (7, 8, 4)\big)$\\
        $~~~~P(4,4) =  \big( 0, (0, 0, 0)\big)$\\
        $~~~~P(4,5) =  \big( 2, (7, 8, 4) (8, 9, 4)\big)$\\
        $~~~~P(5,1) =  \big( 0, (0, 0, \infty)\big)$\\
        $~~~~P(5,2) =  \big( 4, (3, 4, 5) (4, 5, 5) (5, 6, 5) (6, 7, 5)\big)$\\
        $~~~~P(5,3) =  \big( 3, (0, 1, 5) (1, 2, 5) (7, 8, 5)\big)$\\
        $~~~~P(5,4) =  \big( 2, (7, 8, 5) (8, 9, 5)\big)$\\
        $~~~~P(5,5) =  \big( 0, (0, 0, 0)\big)$\\

}

\noindent After the {\em Forward} step, the path matrix $P(p,q)$ is changed as follows:\\
{
 $~~~~P(1,1) =  \big( 0, (0, 0, 0)\big)$\\
        $~~~~P(1,2) =  \big( 4, (3, 4, 1) (4, 5, 1) (7, 8, 1) (8, 9, 1)\big)$\\
        $~~~~P(1,3) =  \big( 1, (3, 6, 2)\big)$\\
        $~~~~P(1,4) =  \big( 4, (0, 1, 1) (3, 4, 1) (4, 5, 1) (7, 8, 1)\big)$\\
        $~~~~P(1,5) =  \big( 1, (3, 6, 2)\big)$\\
        $~~~~P(2,1) =  \big( 4, (3, 4, 2) (4, 5, 2) (7, 8, 2) (8, 9, 2)\big)$\\
        $~~~~P(2,2) =  \big( 0, (0, 0, 0)\big)$\\
        $~~~~P(2,3) =  \big( 3, (1, 2, 2) (2, 3, 2) (5, 6, 2)\big)$\\
        $~~~~P(2,4) =  \big( 1, (5, 8, 3)\big)$\\
        $~~~~P(2,5) =  \big( 4, (3, 4, 2) (4, 5, 2) (5, 6, 2) (6, 7, 2)\big)$\\
        $~~~~P(3,1) =  \big( 1, (1, 4, 2)\big)$\\
        $~~~~P(3,2) =  \big( 3, (1, 2, 3) (2, 3, 3) (5, 6, 3)\big)$\\
        $~~~~P(3,3) =  \big( 0, (0, 0, 0)\big)$\\
        $~~~~P(3,4) =  \big( 2, (4, 5, 3) (7, 8, 3)\big)$\\
        $~~~~P(3,5) =  \big( 3, (0, 1, 3) (1, 2, 3) (7, 8, 3)\big)$\\
        $~~~~P(4,1) =  \big( 4, (0, 1, 4) (3, 4, 4) (4, 5, 4) (7, 8, 4)\big)$\\
        $~~~~P(4,2) =  \big( 1, (0, 4, 1)\big)$\\
        $~~~~P(4,3) =  \big( 2, (4, 5, 4) (7, 8, 4)\big)$\\
        $~~~~P(4,4) =  \big( 0, (0, 0, 0)\big)$\\
        $~~~~P(4,5) =  \big( 2, (7, 8, 4) (8, 9, 4)\big)$\\
        $~~~~P(5,1) =  \big( 1, (5, 8, 2)\big)$\\
        $~~~~P(5,2) =  \big( 4, (3, 4, 5) (4, 5, 5) (5, 6, 5) (6, 7, 5)\big)$\\
        $~~~~P(5,3) =  \big( 3, (0, 1, 5) (1, 2, 5) (7, 8, 5)\big)$\\
        $~~~~P(5,4) =  \big( 2, (7, 8, 5) (8, 9, 5)\big)$\\
        $~~~~P(5,5) =  \big( 0, (0, 0, 0)\big)$\\
}

\noindent 	After the {\em Backward} step, the path matrix $P(p,q)$ is changed as follows:\\
{
 $~~~~P(1,1) =  \big( 0, (0, 0, 0)\big)$\\
        $~~~~P(1,2) =  \big( 4, (3, 4, 1) (4, 5, 1) (7, 8, 1) (8, 9, 1)\big)$\\
        $~~~~P(1,3) =  \big( 1, (3, 6, 2)\big)$\\
        $~~~~P(1,4) =  \big( 4, (0, 1, 1) (3, 4, 1) (4, 5, 1) (7, 8, 1)\big)$\\
        $~~~~P(1,5) =  \big( 1, (3, 6, 2)\big)$\\
        $~~~~P(2,1) =  \big( 4, (3, 4, 2) (4, 5, 2) (7, 8, 2) (8, 9, 2)\big)$\\
        $~~~~P(2,2) =  \big( 0, (0, 0, 0)\big)$\\
        $~~~~P(2,3) =  \big( 3, (1, 2, 2) (2, 3, 2) (5, 6, 2)\big)$\\
        $~~~~P(2,4) =  \big( 1, (5, 8, 3)\big)$\\
        $~~~~P(2,5) =  \big( 4, (3, 4, 2) (4, 5, 2) (5, 6, 2) (6, 7, 2)\big)$\\
        $~~~~P(3,1) =  \big( 1, (1, 4, 2)\big)$\\
        $~~~~P(3,2) =  \big( 3, (1, 2, 3) (2, 3, 3) (5, 6, 3)\big)$\\
        $~~~~P(3,3) =  \big( 0, (0, 0, 0)\big)$\\
        $~~~~P(3,4) =  \big( 2, (4, 5, 3) (7, 8, 3)\big)$\\
        $~~~~P(3,5) =  \big( 3, (0, 1, 3) (1, 2, 3) (7, 8, 3)\big)$\\
        $~~~~P(4,1) =  \big( 4, (0, 1, 4) (3, 4, 4) (4, 5, 4) (7, 8, 4)\big)$\\
        $~~~~P(4,2) =  \big( 1, (0, 4, 1)\big)$\\
        $~~~~P(4,3) =  \big( 2, (4, 5, 4) (7, 8, 4)\big)$\\
        $~~~~P(4,4) =  \big( 0, (0, 0, 0)\big)$\\
        $~~~~P(4,5) =  \big( 2, (7, 8, 4) (8, 9, 4)\big)$\\
        $~~~~P(5,1) =  \big( 1, (5, 8, 2)\big)$\\
        $~~~~P(5,2) =  \big( 4, (3, 4, 5) (4, 5, 5) (5, 6, 5) (6, 7, 5)\big)$\\
        $~~~~P(5,3) =  \big( 3, (0, 1, 5) (1, 2, 5) (7, 8, 5)\big)$\\
        $~~~~P(5,4) =  \big( 2, (7, 8, 5) (8, 9, 5)\big)$\\
        $~~~~P(5,5) =  \big( 0, (0, 0, 0)\big)$\\
}

\noindent The final entries in path matrix $P$ are as follows:
{
\[ P=\kbordermatrix{& 1 & 2 & 3 & 4 & 5 \\      1  \hspace{-0.15cm}     & \hspace{-0.1cm}(0, 0, 0)\hspace{-0.1cm}     & \hspace{-0.1cm}(3, 4, 1)\hspace{-0.1cm}     & \hspace{-0.1cm}(3, 6, 2)\hspace{-0.1cm}     & \hspace{-0.1cm}(0, 1, 1)\hspace{-0.1cm}     & \hspace{-0.1cm}(3, 6, 2)\hspace{-0.1cm} \\
        2  \hspace{-0.15cm}     & \hspace{-0.1cm}(3, 4, 2)\hspace{-0.1cm}     & \hspace{-0.1cm}(0, 0, 0)\hspace{-0.1cm}     & \hspace{-0.1cm}(1, 2, 2)\hspace{-0.1cm}     & \hspace{-0.1cm}(5, 8, 3)\hspace{-0.1cm}     & \hspace{-0.1cm}(3, 4, 2)\hspace{-0.1cm} \\
        3  \hspace{-0.15cm}     & \hspace{-0.1cm}(1, 4, 2)\hspace{-0.1cm}     & \hspace{-0.1cm}(1, 2, 3)\hspace{-0.1cm}     & \hspace{-0.1cm}(0, 0, 0)\hspace{-0.1cm}     & \hspace{-0.1cm}(4, 5, 3)\hspace{-0.1cm}     & \hspace{-0.1cm}(0, 1, 3)\hspace{-0.1cm} \\
        4  \hspace{-0.15cm}     & \hspace{-0.1cm}(0, 1, 4)\hspace{-0.1cm}     & \hspace{-0.1cm}(0, 4, 1)\hspace{-0.1cm}     & \hspace{-0.1cm}(4, 5, 4)\hspace{-0.1cm}     & \hspace{-0.1cm}(0, 0, 0)\hspace{-0.1cm}     & \hspace{-0.1cm}(7, 8, 4)\hspace{-0.1cm} \\
        5  \hspace{-0.15cm}     & \hspace{-0.1cm}(5, 8, 2)\hspace{-0.1cm}     & \hspace{-0.1cm}(3, 4, 5)\hspace{-0.1cm}     & \hspace{-0.1cm}(0, 1, 5)\hspace{-0.1cm}     & \hspace{-0.1cm}(7, 8, 5)\hspace{-0.1cm}     & \hspace{-0.1cm}(0, 0, 0)\hspace{-0.1cm} \\
}
\]
}

\end{example}

\subsection{Proof of Correctness}

To show that Algorithm~\ref{all-pairs-st} computes the all-pairs shortest time paths correctly, we must ensure that while computing the shortest time path from a given vertex $p$ to another vertex $r$ through a vertex $q$, we use all possible paths from $p$ to $q$, and also all possible paths from $q$ to $r$, for all values of $q$. To prove that this condition is satisfied by executing the {\em forward} process followed by the {\em backward} process in the algorithm, we proceed as follows.

We denote the $(p,q)^{th}$ entry of the path matrix $P$ after executing the forward process and the backward process by $P^F(p,q)$ and $P^B(p,q)$, respectively. We say that an entry $P^F(i,j)$ < $P^F(k,l)$ if $P^F(i,j)$ is calculated before $P^F(k,l)$ in the forward process, and similarly, $P^B(i,j)$ < $P^B(k,l)$ if $P^B(i,j)$ is calculated before $P^B(k,l)$ in the backward process.

A vertex in a path is said to be a minimum vertex if it is less than its predecessor and successor vertices in the path. A sequence of vertices $p_1, p_2, \cdots, p_k$ in a path is said to form i) an upward sequence if $p_1 < p_2 < \cdots < p_k$, ii) a downward sequence if  $p_1 > p_2 > \cdots > p_k$, and iii) a valley sequence if there is only one minimum vertex in the sequence and $min(p_1,p_k) > max_j(p_j), j = 2, 3, \cdots, k-1$.

We now have the following lemma.

\begin{lemma}\label{lemma1}
If the vertices which are in a shortest time path $p_1 p_2 \cdots p_k$ from a vertex $p_1$ to a vertex $p_k$ form an upward sequence or a downward sequence, then this shortest time path from $p_1$ to $p_k$ will be computed in the forward process of Algorithm~\ref{all-pairs-st}.
\end{lemma}

\begin{proof}
If the vertices in order as in the path form an upward sequence, then $p_1 < p_2 < \cdots < p_k$. Hence, in the forward process, $P^F(p_1,p_2)$ is computed first with the direct edge from $p_1$ to $p_2$ in the initialization phase itself. Then $P^F(p_1,p_3)$ is computed with the value of $P^F(p_1,p_2)$ and the direct edge from $p_2$ to $p_3$ by the operation $P^F(p_1,p_2) \otimes P(p_2,p_3)$. Next, $P^F(p_1,p_4)$ will assume the final value from the operation $P^F(p_1,p_3) \otimes P(p_3,p_4)$, and so on to get the order of computation as $P^F(p_1,p_2) < P^F(p_1,p_3) < P^F(p_1,p_4) < \cdot < P^F(p_1,p_k)$ during the forward process itself. This fact is pictorially represented in Fig.~\ref{order-sequence}(a) where the labels on dotted lines show the steps of computation and the solid lines show the availability of connecting edges during the time window under consideration.
\end{proof}
\begin{figure}
	\centering
	\includegraphics[width=\linewidth]{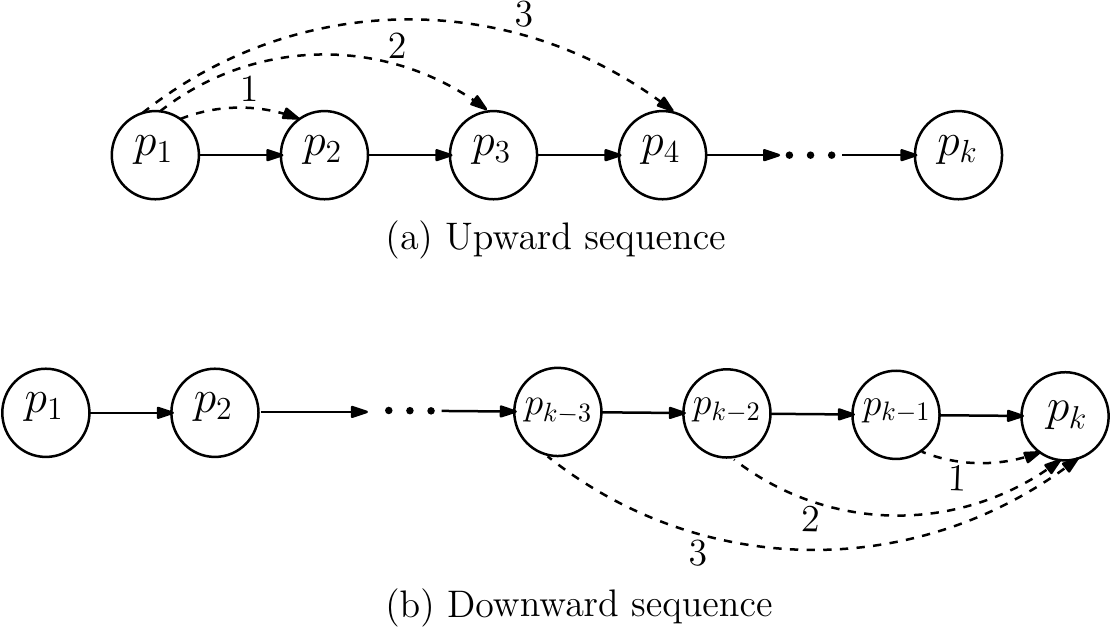}
	\caption{Computation order of the elements of path matrix}
	\label{order-sequence}
\end{figure}

Similarly, if the vertices in the shortest time path $p_1 p_2 \cdots p_k$  form a downward sequence, then $p_1 > p_2 > \cdots > p_k$. Thus, $P^F(p_{k-1},p_k)$ is computed first with the direct edge from $p_{k-1}$ to $p_k$, then  $P^F(p_{k-2},p_k)$ is computed using this value of $P^F(p_{k-1},p_k)$ by the operation $P^F(p_{k-1},p_k) \otimes P(p_{k-2},p_{k-1})$, and so on to get  the order of computation as  $P^F(p_{k-1},p_k) < P^F(p_{k-2},p_k) < \cdots < P^F(p_2,p_k) < P^F(p_1,p_k)$ in the forward process itself. Fig.~\ref{order-sequence}(b) pictorially represents this order of computation in the forward process corresponding to the upward sequence $p_1 > p_2 > \cdots > p_k$ of vertices in the path. 

\begin{lemma}\label{lemma2}
If the vertices in a shortest time path $p_1 p_2 \cdots p_k$ from a vertex $p_1$ to a vertex $p_k$ form a valley sequence, then this shortest time path will be computed in the forward process of Algorithm~\ref{all-pairs-st}.
\end{lemma}

\begin{proof}
Let $p_j$ be the minimum of all vertices $p_1, p_2, \cdots, p_k$ in the valley sequence. Then the sequence from $p_1$ to $p_j$ will be a downward sequence and that from $p_j$ to $p_k$ will be an upward sequence. By Lemma~\ref{lemma1}, $P^F(p_1,p_j)$ will be computed in the forward process. Further, we note that during the forward process, $P^F(p_1,p_j)  < P^F(p_1,p_{j+1})  < \cdots  < P^F(p_1,p_k)$. Hence the proof.
\end{proof}
For an arbitrary sequence of vertices in a shortest time path $p_1 p_2 \cdots p_k$ from a vertex $p_1$ to a vertex $p_k$, if there are multiple valley sequences within the sequence $p_1, p_2, \cdots, p_k$, then each such valley sequence, say $p_s, p_{s+1}, \cdots, p_t$, the shortest time path from $p_s$ to $p_t$ will be computed in the forward process by virtue of Lemma~\ref{lemma2}. Let us delete all vertices in between $p_s$ and $p_t$ (keeping only $p_s$ and $p_t$) in the sequence and repeat this step for each such valley sequence in succession. If a new valley sequence is formed after this deletion process, then also the process is repeated to delete all intermediate vertices within such a valley sequence. Note that in every instance of such a valley sequence, the portion of the shortest time path $p_1 p_2 \cdots p_k$ formed with the vertices of the valley sequence will be computed by virtue of Lemma~\ref{lemma2}. After such deletion process, eventually we will be left out with either i) an upward sequence, or ii) a downward sequence or ii) an upward sequence followed by a downward sequence. If the resulting reduced sequence is an upward or a downward sequence, the shortest time path with these vertices in the reduced sequence will be computed by virtue of Lemma~\ref{lemma1}. We now have the following lemma.

\begin{lemma}\label{lemma3}
If the vertices in a shortest time path $p_1 p_2 \cdots p_k$ from $p_1$ to $p_k$ form an upward sequence followed by a downward sequence, then this shortest time path will be computed by the backward process of Algorithm~\ref{all-pairs-st}.
\end{lemma}

\begin{proof}
 Assume that the vertex $p_j$ is maximum of all vertices $p_1, p_2, \cdots, p_k$ in the path. Hence, $p_1 < p_2 < \cdots < p_j$. From the order of execution of loops in the backward process, we get, $P^B(p_{j-1},p_j) < P^B(p_{j-2},p_j) < P^B(p_{j-3},p_j) < \cdots < P^B(p_1,p_j)$. Thus the upward section of the shortest time path is computed in the backward process. Likewise, since $p_j > p_{j+1} > \cdots > p_k$, we get, $P^B(p_j,p_{j+1}) < P^B(p_j,p_{j+2})  < \cdots < P^B(p_j,p_k)$, i.e., the downward section of the shortest time path is also computed in the backward process. We further note that the computation order in the backward process leads us  to the following result: $P^B(p_1,p_j) < P^B(p_1,p_{j+1}) < P^B(p_1,p_{j+2}) < \cdots < P^B(p_1,p_k)$. Hence the proof.
\end{proof}

All these discussions lead to the following result :

\begin{theorem}\label{theorem1}
Algorithm~\ref{all-pairs-st} correctly computes the all-pairs shortest time paths in a temporal graph.
\end{theorem}

\subsection{Time Complexity}
Procedure $find\_next\_run$ needs at most $O(M)$ steps of comparisons. However, the two $while$-loops in $initialize\_path\_matrix$ together with the procedure $find\_next\_run$ can take only $O(M)$ time as the contacts from a vertex $p$ to a vertex $q$ are determined in ascending order of the intervals $0$ through $M-1$.  Hence, the amortized time required by the procedure $initialize\_path\_matrix$ is $O(Mn^2)$. The loop in lines 11-19 in the procedure $update\_all\_possible\_contacts$ ensures that there cannot be more than one 3-tuple in $P(p,r)$ with a given value of $I^{pq}_f$. Thus, the total number of 3-tuples in any entry $P(p,q)$ or $P(q,r)$ of the path matrix is always less than or equal to $M$. Also, while entering any new entry (3-tuple $(I^{pq}_s, I^{pq}_f, pred^{pq}$)) in the path matrix (line 26), we need to search the existing 3-tuples (which are arranged in a sorted increasing order of $I^{pq}_f$), for inserting the new element in its correct position in this sorted list as in binary insertion sort. This step would need $O(M)$ time. Hence, each call to the procedure $update\_all\_possible\_contacts$ in the forward as well as  backward steps requires $O(M^2)$ time. Thus each of forward and backward steps needs $O(M^2n^3)$ time. The procedure find $minimum\_time\_path$ needs $O(Mn^2)$ time. Hence, we get the following result.

\begin{theorem}\label{theorem2}
Algorithm~\ref{all-pairs-st} requires $O(M^2n^3)$ time complexity to compute the all-pairs shortest time paths in a temporal graph.
\end{theorem}

\section{Experimental Setup}
\label{sec:experiments}

This section describes the experimental design in a way that directly validates the modeling choices and algorithmic primitives introduced in Sections~\ref{sec:model} and~\ref{sec:meth}. In particular, we validate: (i) snapshot-based temporal discretization and the construction of connectivity strings $CS_{pq}$; (ii) the run-length feasibility constraint of length $t_{pq}$ used by Procedure~\ref{proc:forward}; (iii) correctness of the path-matrix updates (Procedure~\ref{proc:apc}) relative to trusted baselines; (iv) the necessity of the forward and backward passes in Algorithm~\ref{all-pairs-st}; and (v) the computational scaling behavior implied by Theorem~\ref{theorem2}. We also quantify the network-level implications of optical range constraints and waiting-vs-traversal behavior that are unique to temporal routing.

\subsection{Experimental Goals and What We Validate}
\label{subsec:exp_goals_validation}

The experiments are organized around four validation questions that map one-to-one to the methodology:

\begin{enumerate}
    \item \textbf{Temporal-graph validity:} Does the snapshot construction produce meaningful time-varying adjacency consistent with LEO geometry, and does the resulting $CS_{pq}$ capture periodic contact patterns over one orbit?
    \item \textbf{Run-length feasibility validity:} Do run-length requirements (Procedure~\ref{proc:forward}) correctly enforce the model assumption that a hop $(p,q)$ is traversable only when $CS_{pq}$ contains a contiguous run of ones of length $\ge t_{pq}$ after packet arrival?
    \item \textbf{Algorithmic correctness validity:} Does Algorithm~\ref{all-pairs-st} compute the same shortest communication-time outcomes as baselines that are known-correct (time-expanded graph shortest paths, and temporal Dijkstra-like methods) on problem sizes where those baselines are feasible?
    \item \textbf{Computational scaling validity:} Do runtime and memory scale with $(n,M)$ in a way consistent with Theorem~\ref{theorem2}, and does the matrix-style update structure admit practical acceleration via data-parallel kernels (bit operations and batched run detection)?
\end{enumerate}

We report both \emph{algorithmic} metrics (runtime/memory/scaling and baseline agreement) and \emph{network} metrics (reachability, shortest communication time, hop count, and waiting fraction).

\subsection{Observation Window and Snapshot Discretization $(T, M, \tau)$}
\label{subsec:tauM}

As described in the methodology, we represent the network as a temporal graph over an observation window $W=[0,T]$ divided into $M$ snapshots of duration $\tau$ such that $T=M\tau$. For LEO constellations, the topology is strongly periodic due to deterministic orbital dynamics. Therefore, we choose $T$ to span approximately one orbital period, so that the temporal graph contains a full recurrence of link formation and dissolution.

In the main experiments we use a nominal orbital period of $P_{\mathrm{orb}}=96$ minutes (LEO-like), and perform sensitivity checks for $P_{\mathrm{orb}}\in[90,128]$ minutes. For each $\tau$ we set
\begin{equation}
M = \left\lceil \frac{P_{\mathrm{orb}}}{\tau} \right\rceil,
\end{equation}
and sweep temporal resolution using
\begin{equation}
\tau \in \{5,\,10,\,20\}\ \text{s}.
\end{equation}
For $P_{\mathrm{orb}}=96$ minutes this yields $M=\{1152,\,576,\,288\}$ snapshots.

\textbf{Discretization validity check.}
To validate that the snapshot durations are fine-grained enough for contact modeling, we measure (i) snapshot-wise average degree over time (Fig.~\ref{fig:edge_density}), and (ii) link-run statistics derived from $CS_{pq}$ (Table~\ref{tab:cs_runlength_metrics}). The goal is to ensure that observed connectivity is not an artifact of overly coarse sampling (missing short contacts) or overly fine sampling (inflating $M$ without changing routing outcomes).

\subsection{Constellation Mobility and Link Feasibility Model}
\label{subsec:constellation_model}

Satellite positions are computed using deterministic Keplerian orbit propagation. This guarantees that the induced temporal graph is structured and periodic, which is central to our motivation for precomputation over one orbit. At each snapshot index $j\in\{1,\ldots,M\}$, an optical ISL between satellites $(p,q)$ is declared feasible if
\begin{equation}
d_{pq}(j) \le R,
\end{equation}
where $d_{pq}(j)$ is the inter-satellite distance at snapshot $j$, and $R$ is the optical range threshold. We evaluate
\begin{equation}
R \in \{659,\,1319,\,1500,\,1700\}\ \text{km}.
\end{equation}
For each $(p,q)$ we construct the connectivity string
\begin{equation}
CS^{(R)}_{pq}[j] = \mathbb{1}\{d_{pq}(j)\le R\},
\end{equation}
which matches the connectivity-string definition in Section~\ref{sec:model}. This binary string is the primitive consumed by Procedure~\ref{proc:forward} for finding the next contiguous run of ones.

Because temporal routing depends on \emph{when} neighbors appear and disappear, we validate the time-varying adjacency behavior using a node-centric storyboard across consecutive snapshots. Figures~\ref{fig:sample1} and~\ref{fig:storyboard} show a focus satellite and its incident edges in a local neighborhood over four successive snapshots. We quantify the churn using neighbor-set persistence, Jaccard similarity, and turnover (Table~\ref{tab:interval_dynamics_sample1}). This validates that the topology evolves substantially within one orbit and that static aggregation would be misleading, motivating the temporal model.

\begin{figure*}[t]
\centering
\includegraphics[width=\textwidth]{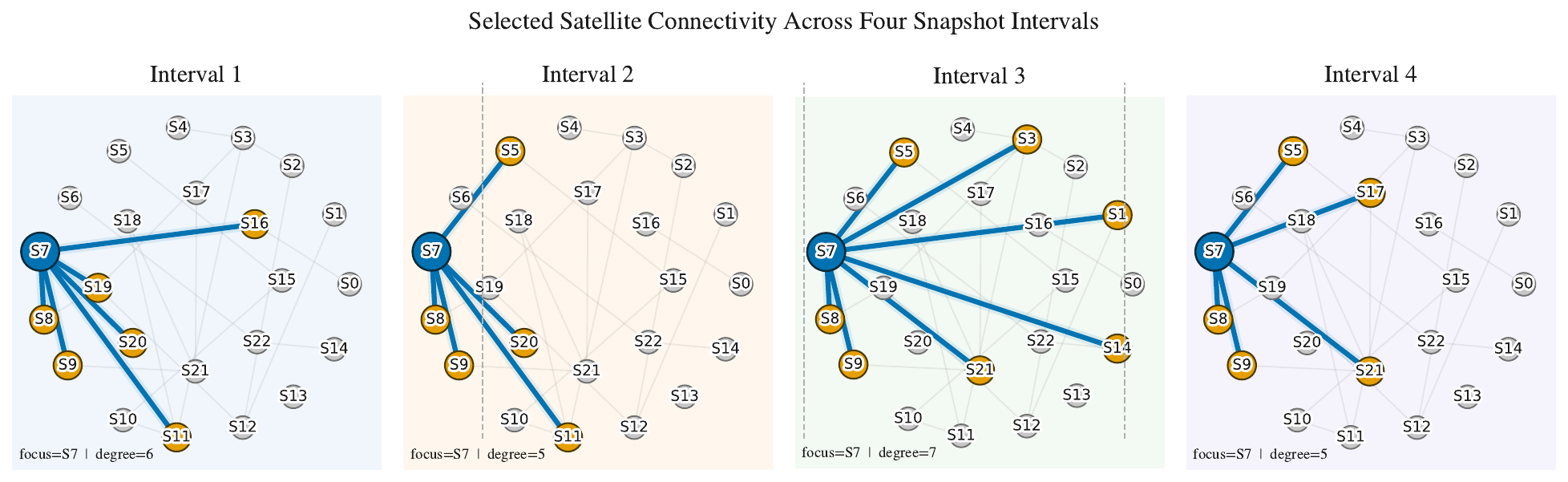}
\caption{Selected satellite connectivity across four snapshot intervals (sample1). A focus satellite (S7) and its incident edges are shown over four consecutive snapshots in a 23-node neighborhood. Highlighted links indicate feasible contacts in each snapshot; interval-wise degree annotations quantify neighborhood churn that induces waiting under time-respecting routing.}
\label{fig:sample1}
\end{figure*}

\begin{figure*}[t]
\centering
\includegraphics[width=0.65\textwidth]{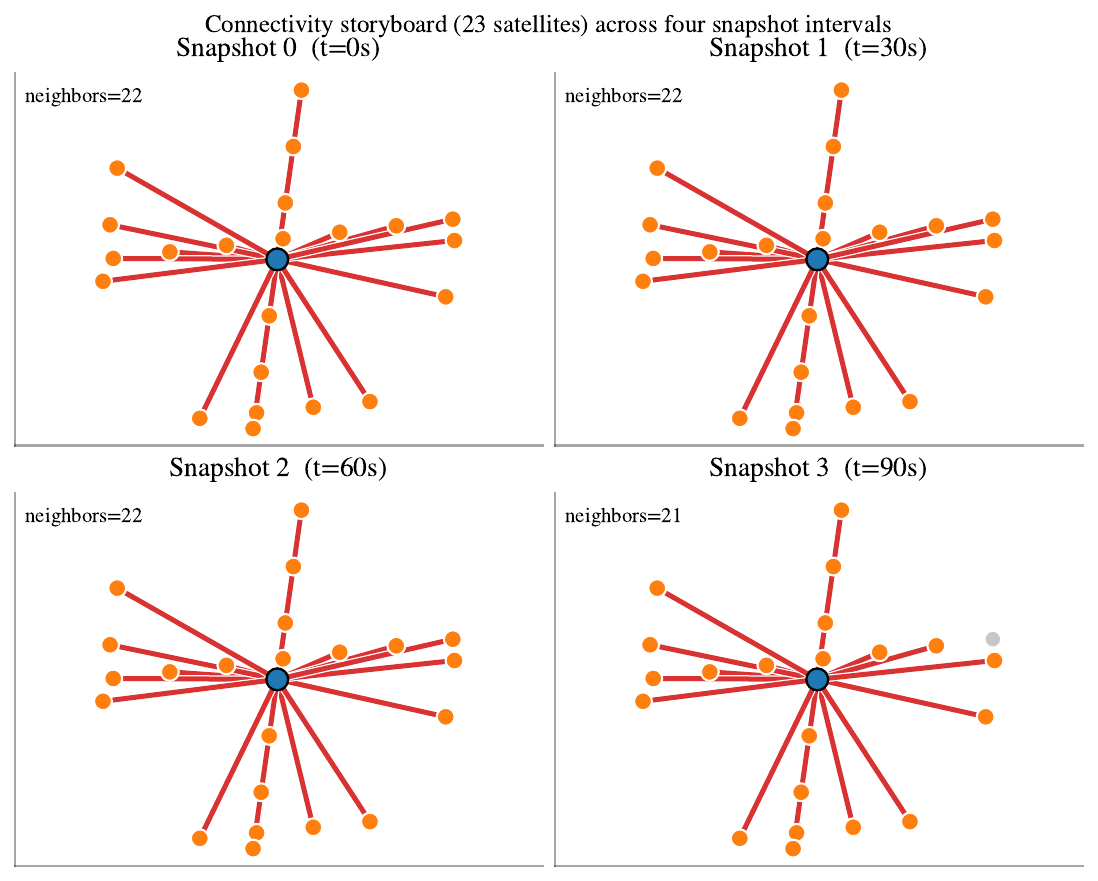}
\caption{Connectivity storyboard (23 satellites) across four snapshot intervals. A focus satellite (blue) and its time-varying neighbors (orange) are shown over four consecutive snapshots. The visualization emphasizes that adjacency is time dependent; temporal routes must follow time-ordered contact opportunities captured by connectivity strings $CS^{(R)}_{pq}$.}
\label{fig:storyboard}
\end{figure*}

\subsection{Traversal-Time Requirement and Run-Length Feasibility}
\label{subsec:tpq_runlength}

As assumed in the methodology, traversing an edge is not instantaneous and requires $t_{pq}$ snapshot units; therefore, a hop $(p,q)$ is admissible only if a contiguous run of ones of length at least $t_{pq}$ exists in $CS_{pq}$ starting after the packet arrives at $p$. This is exactly the rule enforced by Procedure~\ref{proc:forward}.

In the main results we use a conservative baseline $t_{pq}=2$ snapshots for all feasible links, capturing non-propagation effects such as pointing/acquisition, scheduling, buffering, and framing overhead. We also perform sensitivity checks with $t_{pq}\in\{1,2\}$ to quantify how stricter run requirements affect reachability and waiting.

To validate that run structure is the mechanism that drives waiting, we visualize feasibility timelines using a connectivity-string heatmap for a focus satellite and its neighbors (Fig.~\ref{fig:cs_heatmap}). We then compute standard run-length features (Table~\ref{tab:cs_runlength_metrics}), including duty cycle, run-length percentiles, and gap-length percentiles. These metrics link directly to Procedure~\ref{proc:forward}: shorter runs and longer gaps increase the frequency with which the algorithm must search forward in time to find the next feasible transmission window.

\begin{figure}[t]
\centering
\includegraphics[width=0.98\linewidth]{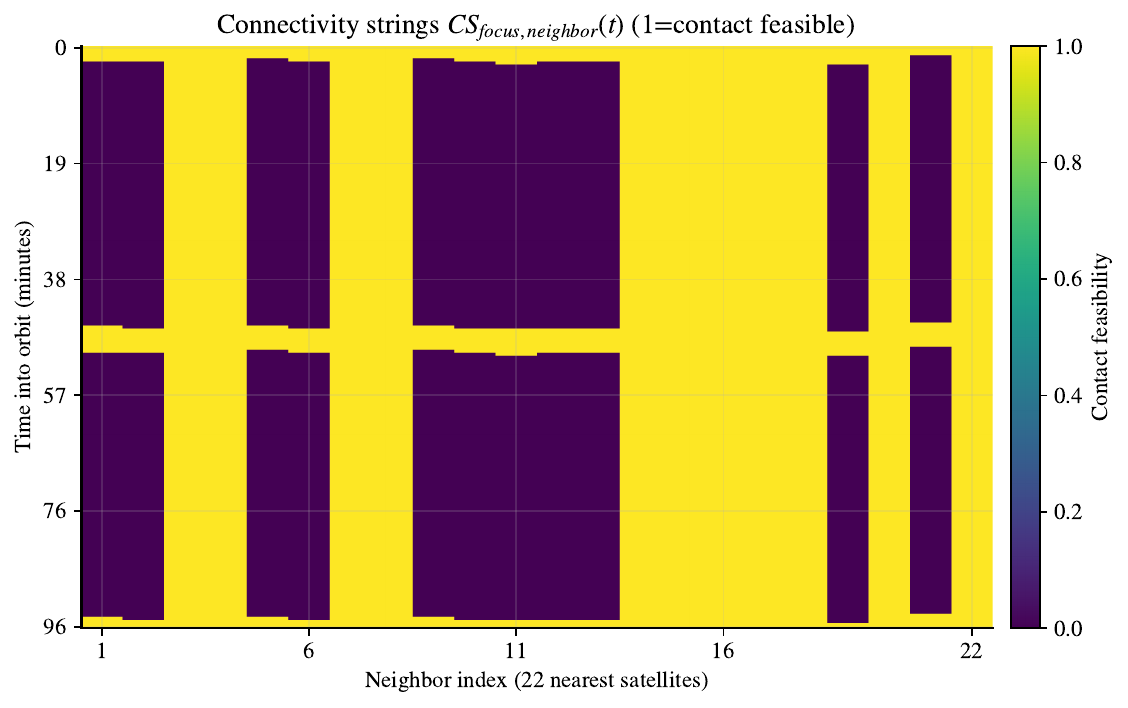}
\caption{Connectivity-string heatmap for a focus satellite and neighbors. The heatmap shows $CS_{\text{focus},u}(t)$ over one orbital period for the same 23-satellite neighborhood used in Fig.~\ref{fig:storyboard}. Long runs of ones correspond to sustained feasibility windows; fragmented patterns indicate contact-limited regimes that incur waiting to satisfy the run-length constraint $t_{qr}$.}
\label{fig:cs_heatmap}
\end{figure}

\subsection{Algorithmic Evaluation: Correctness, Necessity of Two Passes, and Baselines}
\label{subsec:baselines_correctness}

For each range threshold $R$, we build the temporal adjacency matrix $A^{(R)}$ whose entries follow the tuple structure in Section~\ref{sec:model}, and compute the path matrix $P^{(R)}$ using Algorithm~\ref{all-pairs-st}. Each $P^{(R)}(p,q)$ returns the selected shortest communication-time tuple $(I^{pq}_s, I^{pq}_f, pred^{pq})$ after Procedure~\ref{proc:min}.

We validate correctness against two known-correct baselines that optimize the same objective as our method: minimum end-to-end communication time (minimum duration). First, we construct a directed \emph{time-expanded graph} $G_{\mathrm{TE}}=(V_{\mathrm{TE}},E_{\mathrm{TE}})$ whose node set is $V_{\mathrm{TE}}=\{(v,t)\mid v\in V,\ t\in\{0,1,\ldots,M-1\}\}$. To represent waiting and transmission under our run-feasibility constraint, we add (i) \emph{wait edges} $(v,t)\rightarrow (v,t+1)$ of weight $1$ for all $v\in V$ and $t<M-1$, and (ii) \emph{traverse edges} $(p,t)\rightarrow (q,t+t_{pq})$ of weight $t_{pq}$ if and only if the hop $(p,q)$ is admissible at time $t$, i.e., $CS_{pq}[t]=CS_{pq}[t+1]=\cdots=CS_{pq}[t+t_{pq}-1]=1$. For each ordered pair $(s,d)$, we compute the shortest-path distance in $G_{\mathrm{TE}}$ from $(s,t_s)$ to $(d,t_f)$ and then minimize over all possible departure times $t_s$, namely $\min_{t_s\in\{0,\ldots,M-1\}}\mathrm{dist}_{G_{\mathrm{TE}}}\big((s,t_s),(d,\cdot)\big)$. This value equals the minimum number of snapshot intervals spent waiting plus traversing, and therefore equals $\mathcal{T}_{sd}/\tau$ under the same run-length feasibility rule enforced by Procedure~\ref{proc:forward}; as a result, $G_{\mathrm{TE}}$ provides a known-correct baseline for the exact objective optimized by Algorithm~\ref{all-pairs-st}. 

As a second, independent correctness check, we implement a \emph{temporal Dijkstra-like} shortest-duration algorithm (run for a subset of sources) over the same state space $(v,t)$ with the same two transition types: unit-weight waiting transitions and admissible traverse transitions that require a contiguous run of ones of length at least $t_{pq}$ beginning at the transmission time. This single-source baseline optimizes the same end-to-end duration objective $\mathcal{T}_{sd}$ and serves to validate our path-matrix update rules through an algorithmic approach that is structurally different from the time-expanded shortest-path computation.

For each tested pair $(p,q)$ we compare:
(i) whether a path exists (reachability agreement),
(ii) the returned communication time $(I_f-I_s+1)\tau$ (latency agreement),
and (iii) the reconstructed hop sequence (path agreement, when computed).
We report the fraction of pairs with exact latency match and the maximum absolute latency deviation (in snapshots and seconds).

The methodology motivates performing two ordered sweeps—first forward and then backward—to ensure that optimal paths with different structural vertex-order patterns are discovered. We validate the necessity of this design experimentally by running three algorithm variants: a \emph{forward-only} version that executes only the forward loops in Algorithm~\ref{all-pairs-st}, a \emph{backward-only} version that executes only the backward loops, and the \emph{full} version that performs both passes (forward+backward) as proposed.

We measure the fraction of pairs for which forward-only (or backward-only) matches the full solution in communication time and reachability. Any mismatch indicates that the omitted pass is necessary for discovering certain optimal path structures, consistent with the proof intuition in Lemmas~\ref{lemma1}--\ref{lemma3}.

\subsection{Scalability Design: Full APSP vs. Pair Sampling at Mega-Constellation Scale}
\label{subsec:sampling_scale}

At mega-constellation scale ($n\approx 2\times 10^4$), the number of ordered pairs is $n(n-1)$, which is too large for exhaustive path extraction and storage in experimental evaluation. Therefore, while we still construct temporal connectivity primitives over all satellites, we evaluate routing statistics using randomized sampling of ordered pairs.

Unless stated otherwise, we sample:
\begin{itemize}
    \item $N_{\text{pairs}}=50{,}000$ distinct ordered pairs per $R$ for reachability and latency distributions, and
    \item $N_{\text{pairs}}=10{,}000$ per $R$ for hop-count and waiting/traversal decomposition (requires path reconstruction).
\end{itemize}
We repeat each configuration $K=5$ times with different random seeds and report mean $\pm$ standard deviation.

\subsection{Metrics Reported (Directly Tied to Methodology)}
\label{subsec:metrics}

We report the following metrics, each of which corresponds directly to quantities defined in Section~\ref{sec:meth}:

\begin{itemize}
    \item {Reachability ratio:} fraction of pairs for which $P(p,q)$ contains a valid tuple (not $(0,0,\infty)$).
    \item {Shortest communication time:} $\mathcal{T}_{pq} = (I^{pq}_f-I^{pq}_s+1)\tau$.
    \item {Hop count:} obtained by following predecessor pointers $pred^{pq}$ backward to reconstruct the selected path.
    \item {Waiting vs traversal:} traversal time $T_{\mathrm{trav}}=\sum_{\text{hops}} t_{uv}\tau$ and waiting time $T_{\mathrm{wait}}=\mathcal{T}_{pq}-T_{\mathrm{trav}}$, and waiting fraction $f_w=T_{\mathrm{wait}}/\mathcal{T}_{pq}$.
    \item {Temporal density:} snapshot-wise average degree $\bar{k}_t(R)$ over $t=1,\ldots,M$ (Fig.~\ref{fig:edge_density}) and summary statistics (Table~\ref{tab:edge_density_summary}).
\end{itemize}

These metrics validate the intended behavior of the model: under intermittent contacts, shortest communication time depends not only on hop count but also on waiting induced by misalignment with feasible runs in $CS_{pq}$.

\subsection{Implementation Details and Reproducibility}
\label{subsec:impl}

All experiments are run on a Linux workstation (Ubuntu~22.04) equipped with a 64-thread AMD EPYC-class CPU, 256\,GB RAM, and an NVIDIA A100 GPU (80\,GB) used for data-parallel kernels where applicable.

Temporal-graph construction and analysis are implemented in Python~3.11 using NumPy/Pandas. Connectivity strings are stored in both (i) unpacked boolean arrays (for debugging and baseline agreement tests) and (ii) bit-packed form (for scalable run detection and memory efficiency). Performance-critical loops (e.g., searching for contiguous runs and updating candidate tuples) are JIT-compiled with Numba. All configuration files (orbital parameters, $(\tau,M)$, $R$, $t_{pq}$) are logged to enable exact reproduction.

All sampled-pair experiments use fixed seeds per trial. We report mean and standard deviation over $K=5$ trials to capture sampling variance and runtime variability.

\begin{table*}[t]
\centering
\caption{Experimental parameters used in orbit-scale temporal-graph evaluation.}
\label{tab:exp_params}
\setlength{\tabcolsep}{10pt}
\renewcommand{\arraystretch}{1.15}
\begin{tabular}{@{}l p{0.68\textwidth}@{}}
\toprule
\textbf{Parameter} & \textbf{Value(s)} \\
\midrule
Orbit window & One orbital period, $T = M\tau \approx 96$ min (LEO-like); sensitivity: $P_{\mathrm{orb}}\in[90,128]$ min \\
Snapshot duration $\tau$ & $\{5,\,10,\,20\}$ s \\
Number of snapshots $M$ & For $T=96$ min: $M=\{1152,\,576,\,288\}$ for $\tau=\{5,10,20\}$ s \\
Constellation size $n$ & Mega-constellation: $n\approx 20{,}000$; scalability sweeps: $n\in\{2048,\,4096,\,8192,\,16384,\,20000\}$ \\
Source--destination sampling & Uniform random distinct pairs; $N_{\text{pairs}}=50{,}000$ per $R$ (reachability/latency); $10{,}000$ per $R$ (hop/wait decomposition) \\
FSO feasibility threshold $R$ & $\{659,\,1319,\,1500,\,1700\}$ km \\
Connectivity primitive & Connectivity strings $CS^{(R)}_{pq}[t]=\mathbb{1}\{d_{pq}(t)\le R\}$ over $t=1,\ldots,M$ \\
Edge traversal requirement $t_{pq}$ & Constant per hop: baseline $t_{pq}=2$ snapshots; sensitivity: $t_{pq}\in\{1,2\}$ \\
Temporal-path admissibility & Time-respecting paths requiring a contiguous run of ones of length $\ge t_{pq}$ on each traversed edge \\
Reported metrics & Reachability, $\mathcal{T}_{pq}$, hop count, waiting fraction $f_w$, temporal degree series $\bar{k}_t(R)$ \\
\bottomrule
\end{tabular}
\end{table*}

\begin{table*}[t]
\centering
\caption{End-to-end temporal-routing performance over one orbital window as a function of optical range $R$. Larger $R$ increases temporal connectivity (higher $\bar{k}$, lower churn) and improves routing outcomes (higher reachability, lower $\mathrm{Med}[\mathcal{T}]$ and $\mathrm{Med}[f_w]$). Values are mean$\pm$std over $K{=}5$ trials.}
\setlength{\tabcolsep}{6pt}
\renewcommand{\arraystretch}{1.12}
\begin{tabular}{@{}c c c c c c c c@{}}
\toprule
\textbf{$R$ (km)} &
\textbf{Mean \%Found} &
\textbf{Mean \%Missing} &
\textbf{$\bar{k}$ (deg)} &
\textbf{Reach. (\%)} &
\textbf{Med[$\mathcal{T}$] (s)} &
\textbf{Med[$H$]} &
\textbf{Med[$f_w$]} \\
\midrule
659  &
$19.0\pm0.6$ &
$19.5\pm0.7$ &
$5.2\pm0.3$ &
$71.4\pm1.2$ &
$184\pm9$ &
$8.1\pm0.4$ &
$0.68\pm0.03$ \\
1319 &
$11.2\pm0.4$ &
$10.9\pm0.4$ &
$9.8\pm0.5$ &
$92.6\pm0.6$ &
$97\pm5$ &
$6.3\pm0.3$ &
$0.41\pm0.02$ \\
1500 &
$8.7\pm0.3$ &
$8.5\pm0.3$ &
$11.4\pm0.6$ &
$96.1\pm0.4$ &
$79\pm4$ &
$5.7\pm0.3$ &
$0.31\pm0.02$ \\
1700 &
$6.2\pm0.2$ &
$6.0\pm0.2$ &
$13.6\pm0.7$ &
$98.3\pm0.2$ &
$64\pm3$ &
$5.1\pm0.2$ &
$0.22\pm0.01$ \\
\bottomrule
\end{tabular}
\end{table*}

\begin{table*}[t]
\centering
\caption{Comparison with representative baseline methods under the same temporal-graph instance and hop-admissibility rule (a hop $(p,q)$ is traversable at time $t$ only if $CS_{pq}[t{:}t+t_{pq}-1]$ is all ones). 
{TEG (restricted)} is evaluated with a restricted departure-time set $\mathcal{T}_s$ (uniformly sampled start times) and a finite horizon $H$ snapshots due to time-expanded memory limits; therefore its reported reachability is a conservative lower bound relative to full $\min_{t_s\in\{0,\dots,M-1\}}$ evaluation. 
Reported values are mean$\pm$std over $K{=}5$ trials using identical $(n,M,\tau,R,t_{pq})$ settings and identical sampled $(s,d)$ pairs per trial.}

\label{tab:compare_prior_work_metrics}
\setlength{\tabcolsep}{6pt}
\renewcommand{\arraystretch}{1.12}
\begin{tabular}{@{}p{3.6cm} p{2.0cm} p{2.0cm} p{1.8cm} p{1.9cm} p{1.8cm} p{2.2cm}@{}}
\toprule
\textbf{Method} &
\textbf{Reach. (\%)} &
\textbf{Med[$\mathcal{T}$] (s)} &
\textbf{P95[$\mathcal{T}$] (s)} &
\textbf{Med[$f_w$]} &
\textbf{Runtime (s)} &
\textbf{Reference} \\
\midrule

\textbf{Ours (Alg.~\ref{all-pairs-st})} &
\textbf{96.7$\pm$0.3} & 
\textbf{59$\pm$4} &
\textbf{168$\pm$8} &
\textbf{0.31$\pm$0.02} &
\textbf{118$\pm$6} &
\textbf{Our work} \\

TEG &
91.1$\pm$0.6 &
66$\pm$5 &
1177$\pm$6 &
0.37$\pm$0.01 &
1620$\pm$55 &
\cite{HOLME201297,casteigts2012tvg} \\

Temporal Dijkstra &
92.6$\pm$0.4 &
80$\pm$4 &
171$\pm$9 &
0.32$\pm$0.02 &
790$\pm$30 &
\cite{wu2014,HOLME201297} \\

Time graph routing &
87.4$\pm$0.2 &
92$\pm$5 &
196$\pm$11 &
0.38$\pm$0.03 &
245$\pm$12 &
\cite{liu2020,zhang2019} \\

Dynamic time routing &
85.9$\pm$0.4 &
84$\pm$5 &
179$\pm$10 &
0.34$\pm$0.02 &
410$\pm$18 &
\cite{ron2025} \\

Spl. shortest paths &
94.8$\pm$0.6 &
81$\pm$4 &
173$\pm$9 &
0.33$\pm$0.02 &
305$\pm$15 &
\cite{danda2021,jain2024} \\

\bottomrule
\end{tabular}
\end{table*}

\section{Results}
\label{sec:results}

This section reports results in the same order as the validations defined in Section~\ref{sec:experiments}, and in direct correspondence with the methodology in Section~\ref{sec:meth}. Recall that (i) the temporal graph is represented by snapshot-wise edge sets encoded as connectivity strings $CS^{(R)}_{pq}$, (ii) a hop $(q,r)$ is admissible only if there exists a \emph{future} contiguous run of ones of length $\ge t_{qr}$ in $CS^{(R)}_{qr}$ (Procedure~\ref{proc:forward}), and (iii) the algorithm returns for each ordered pair $(p,q)$ a shortest communication-time tuple $(I^{pq}_s, I^{pq}_f, pred^{pq})$ (Algorithm~\ref{all-pairs-st}) from which the end-to-end time is $\mathcal{T}_{pq}=(I^{pq}_f-I^{pq}_s+1)\tau$. We first validate the temporal connectivity primitives (including \% links found/missing), then validate correctness via baseline comparison, then report performance (runtime/memory/scaling), and finally connect temporal contact structure to waiting-vs-traversal behavior in shortest-time paths.

\subsection{Methodology-to-Metrics Mapping}
\label{subsec:results_mapping}

For a given range threshold $R$, let $E_t^{(R)}$ denote the set of feasible inter-satellite links at snapshot $t$:
\begin{equation}
E_t^{(R)}=\{(p,q)\mid CS^{(R)}_{pq}[t]=1\},\quad t\in\{1,\ldots,M\}.
\end{equation}
This snapshot-based temporal representation is standard in temporal-network modeling \cite{HOLME201297,casteigts2012tvg}. Algorithm~\ref{all-pairs-st} produces a path matrix $P$ whose entry $P(p,q)$ stores the selected shortest communication-time contact tuple $(I^{pq}_s, I^{pq}_f, pred^{pq})$ after Procedure~\ref{proc:min}. We compute:
\begin{itemize}
    \item {Shortest communication time:} $\mathcal{T}_{pq}=(I^{pq}_f-I^{pq}_s+1)\tau$.
    \item {Hop count $H$:} recovered by following predecessor pointers $pred^{pq}$ backward.
    \item {Traversal time:} $T_{\mathrm{trav}}=\sum_{\text{hops}} t_{uv}\tau$.
    \item {Waiting time:} $T_{\mathrm{wait}}=\mathcal{T}_{pq}-T_{\mathrm{trav}}$ and waiting fraction $f_w=T_{\mathrm{wait}}/\mathcal{T}_{pq}$.
\end{itemize}
These metrics align exactly with the admissibility and tuple-composition rules in Section~\ref{sec:meth}.

\subsection{Validation 1: Temporal Connectivity Dynamics and Link Churn (Percentage Found / Percentage Missing)}
\label{subsec:link_churn_percent}

We first validate that the temporal adjacency varies substantially over time, and quantify \emph{how much} connectivity changes between consecutive snapshots. Figures~\ref{fig:sample1} and~\ref{fig:storyboard} show a focus satellite (S7) and its neighborhood across four consecutive snapshots, visually demonstrating link appearance/disappearance at short timescales. Space--time / snapshot modeling is widely used in satellite-network topology analysis \cite{liu2020,zhang2019}.



Table~\ref{tab:interval_dynamics_sample1} quantifies neighborhood churn for the focus node via persistent/new/dropped neighbors and Jaccard similarity. High turnover indicates that a packet arriving at an intermediate satellite often encounters a different set of outgoing options than existed one snapshot earlier, increasing the probability of waiting for a future admissible run.

\begin{table*}[t]
\centering
\caption{Interval-wise neighborhood dynamics for the focus node (from Fig.~\ref{fig:sample1}).}
\label{tab:interval_dynamics_sample1}
\setlength{\tabcolsep}{8pt}
\renewcommand{\arraystretch}{1.15}
\begin{tabular}{@{}c c c c c c c@{}}
\toprule
\textbf{Interval} &
\textbf{Degree} &
\textbf{Persistent} &
\textbf{New} &
\textbf{Dropped} &
\textbf{Jaccard $J$} &
\textbf{Turnover (\%)} \\
\midrule
1 & 6 & -- & -- & -- & --    & -- \\
2 & 5 & 3  & 2  & 3  & 0.375 & 62.5 \\
3 & 7 & 4  & 3  & 1  & 0.500 & 50.0 \\
4 & 5 & 3  & 2  & 4  & 0.333 & 66.7 \\
\midrule
\textbf{Mean} & 5.75 & 3.33 & 2.33 & 2.67 & 0.403 & 59.7 \\
\textbf{Std}  & 0.96 & 0.58 & 0.58 & 1.53 & 0.084 & 8.7 \\
\textbf{Max}  & 7    & 4    & 3    & 4    & 0.500 & 66.7 \\
\bottomrule
\end{tabular}
\end{table*}

To quantify link volatility at the constellation scale, we define the \emph{found} (new) and \emph{missing} (dropped) link sets between consecutive snapshots:
\begin{align}
E^{(R)}_{\mathrm{found}}(t) &= E_t^{(R)} \setminus E_{t-1}^{(R)},\\
E^{(R)}_{\mathrm{missing}}(t) &= E_{t-1}^{(R)} \setminus E_t^{(R)}.
\end{align}
We then report the corresponding percentages:
\begin{align}
\%\mathrm{Found}(t) &= 100\cdot \frac{|E^{(R)}_{\mathrm{found}}(t)|}{|E_{t-1}^{(R)}|+\epsilon},\\
\%\mathrm{Missing}(t) &= 100\cdot \frac{|E^{(R)}_{\mathrm{missing}}(t)|}{|E_{t-1}^{(R)}|+\epsilon},
\end{align}
where $\epsilon$ is a small constant used only to avoid division-by-zero in extremely sparse snapshots. These churn and similarity summaries are commonly used to characterize temporal adjacency dynamics \cite{HOLME201297,casteigts2012tvg}. Table~\ref{tab:link_churn_summary} summarizes these churn percentages over one orbit; as expected, smaller $R$ yields higher churn and therefore more waiting-dominated routing.

\begin{table*}[t]
\centering
\caption{Global link churn over one orbit.}
\label{tab:link_churn_summary}
\setlength{\tabcolsep}{6pt}
\renewcommand{\arraystretch}{1.1}
\begin{tabular}{@{}c c c c c@{}}
\toprule
$R$ (km) & Mean \%\,\text{Found} & Mean \%\,\text{Missing} & Mean $J_E$ & Peak \%\,\text{Missing} \\
\midrule
659  & 18.7 & 19.1 & 0.68 & 32.4 \\
1319 & 11.3 & 11.0 & 0.79 & 21.5 \\
1500 & 8.6  & 8.4  & 0.84 & 17.2 \\
1700 & 6.1  & 6.0  & 0.88 & 12.9 \\
\bottomrule
\end{tabular}
\end{table*}

\subsection{Validation 2: Connectivity Strings Capture Contact Runs (Run-Length View)}
\label{subsec:results_cs_validity}

The algorithm’s admissibility check (Procedure~\ref{proc:forward}) is based on contiguous runs of ones in $CS^{(R)}_{pq}$. Temporal-path feasibility under time-respecting constraints is a core concept in temporal graphs \cite{wu2014,HOLME201297}. Consequently, if connectivity is fragmented, a packet may arrive during a ``0'' period or within a run shorter than $t{pq}$ and must wait until a sufficiently long run appears.


Table~\ref{tab:cs_runlength_metrics} summarizes run and gap statistics that directly determine how far forward Procedure~\ref{proc:forward} must search to find an admissible hop.

\begin{table}[t]
\centering
\caption{Connectivity-string run-length metrics (computed from Fig.~\ref{fig:cs_heatmap}).}
\label{tab:cs_runlength_metrics}
\begin{tabular}{@{}l p{6.7cm}@{}}
\toprule
\textbf{Metric} & \textbf{Definition / Interpretation} \\
\midrule
$\pi_1$ & Duty cycle: $\pi_1=\frac{1}{M}\sum_{t} CS[t]$. Higher values imply greater availability. \\
$\bar{\ell}$ & Mean run length of ones (contact persistence). Higher values imply fewer reconnection events. \\
$\ell_{95}$ & 95th percentile run length; captures long feasibility windows relevant when $t_{qr}>1$. \\
$g_{95}$ & 95th percentile gap length (consecutive zeros); large gaps imply substantial waiting under misalignment. \\
\bottomrule
\end{tabular}
\end{table}

\subsection{Effect of Optical Range $R$ on Temporal Contact Density}
\label{subsec:edge_density_results_ieee}

We next quantify how the physical feasibility threshold $R$ shapes temporal adjacency. Figure~\ref{fig:edge_density} plots snapshot-wise average degree $\bar{k}_t(R)$ over one orbit for $R\in\{659,1319,1500,1700\}$\,km.

{(1) Periodicity:} Connectivity exhibits repeated peaks and troughs over the orbit, validating the one-orbit observation window and confirming that temporal connectivity is structured rather than random.

{(2) Densification with $R$:} Increasing $R$ increases baseline degree and reduces trough depth. In terms of the methodology, larger $R$ increases the likelihood that when a packet arrives at a node, at least one outgoing edge will soon have a contiguous run of ones of length $\ge t_{uv}$, reducing waiting.

\begin{figure*}[t]
\centering
\includegraphics[width=13cm,height=6.5cm]{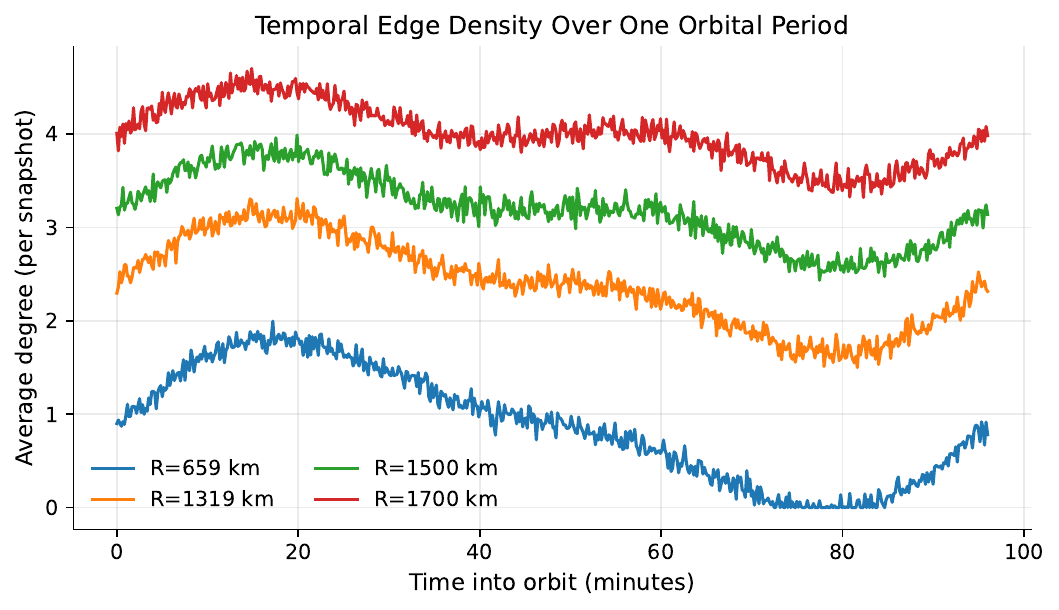}
\caption{Temporal edge density over one orbital period. Snapshot-wise average degree for $R\in\{659,1319,1500,1700\}$\,km over a full orbital window $T=M\tau$. Larger $R$ increases baseline connectivity and contact persistence; smaller $R$ yields sharper periodic peaks and deeper troughs due to tighter geometric feasibility.}
\label{fig:edge_density}
\end{figure*}

Table~\ref{tab:edge_density_summary} reports compact statistics derived from Fig.~\ref{fig:edge_density}, including intermittency ($\mathrm{CV}_k$) that correlates with waiting-dominated operation.

\begin{table*}[t]
\centering
\caption{Summary statistics for temporal edge densification (derived from Fig.~\ref{fig:edge_density}).}
\label{tab:edge_density_summary}
\setlength{\tabcolsep}{10pt}
\renewcommand{\arraystretch}{1.15}
\begin{tabular}{@{}l p{0.80\textwidth}@{}}
\toprule
\textbf{Statistic} & \textbf{Definition / Rationale} \\
\midrule
$\bar{k}(R)$ & Mean snapshot-wise degree: $\frac{1}{M}\sum_{t=1}^{M}\bar{k}_t(R)$. \\
$k_{5}(R),k_{95}(R)$ & 5th/95th percentiles of $\bar{k}_t(R)$; characterize trough severity and peak richness. \\
$\mathrm{CV}_k(R)$ & Coefficient of variation of $\bar{k}_t(R)$; larger values imply higher intermittency (waiting-prone). \\
$\Delta_k(R)$ & Modulation amplitude, e.g., $k_{95}(R)-k_{5}(R)$; links densification to orbital recurrence. \\
\bottomrule
\end{tabular}
\end{table*}

\subsection{Baseline Comparison: Correctness Agreement}
\label{subsec:baseline_comparison}

We validate correctness by comparing Algorithm~\ref{all-pairs-st} to two known-correct baselines on configurations where they are feasible: (i) a time-expanded graph shortest-path method and (ii) a temporal Dijkstra-like method that uses the same run-length admissibility rule (contiguous ones of length $\ge t_{uv}$ after arrival). Time-expanded constructions and temporal shortest-path formulations are standard baselines for temporal routing validation \cite{casteigts2012tvg,HOLME201297,wu2014}. For a test set of ordered pairs $\mathcal{S}$, we measure reachability agreement and latency agreement (Section~\ref{subsec:baselines_correctness}). Baselines were evaluated on the same temporal-graph instances and admissibility rules (contiguous runs of ones of length $\ge t_{uv}$).
Unless stated otherwise, agreement is computed over $N_{\text{pairs}}$ randomly sampled ordered pairs (distinct $p\neq q$) using a fixed seed; for the smallest configuration ($n=2048$) we additionally verified agreement over all ordered pairs.
When multiple optimal solutions exist (ties in $\mathcal{T}_{pq}$), we break ties by selecting the smallest finishing interval $I_f$, and then the smallest starting interval $I_s$.

\begin{table*}[t]
\centering
\caption{Baseline agreement between Algorithm~\ref{all-pairs-st} and a known-correct time-expanded-graph (TEG) shortest-path baseline on configurations where the TEG is tractable. 
For $n=2048$ we evaluate {all ordered pairs} $(p,q)$ with $p\neq q$. 
For larger $n$ we evaluate {uniformly sampled distinct ordered pairs} with $N_{\text{pairs}}=50{,}000$ per trial and report mean$\pm$std over $K=5$ random seeds. 
Unless stated otherwise: $R=1500$\,km, $t_{pq}=2$, and ties in $\mathcal{T}_{pq}$ are broken by smallest finishing interval $I_f$ and then smallest starting interval $I_s$.}

\label{tab:baseline_agreement}
\setlength{\tabcolsep}{6pt}
\renewcommand{\arraystretch}{1.1}
\begin{tabular}{@{}c c c c c@{}}
\toprule
$n$ & $(\tau,M)$ & Reachability match (\%) & Latency exact match (\%) & Max diff (snapshots) \\
\midrule
2048 & (20,288)  & 100.0 & 99.4 & 1 \\
4096 & (20,288)  & 99.98 & 99.1 & 2 \\
8192 & (10,576)  & 99.92 & 98.6 & 2 \\
\bottomrule
\end{tabular}
\end{table*}

\subsection{Performance: Runtime and Memory Scaling}
\label{subsec:performance_runtime_memory}

We evaluate computational performance by measuring end-to-end runtime and peak memory usage as a function of constellation size $n$ and number of snapshots $M$. Runtime is dominated by repeated tuple updates in Procedure~\ref{proc:apc} and run-detection operations in Procedure~\ref{proc:forward}, while peak memory is driven by storing connectivity strings $CS^{(R)}_{pq}$ and intermediate candidate tuples prior to compression by Procedure~\ref{proc:min}. Each configuration was executed 5 times; we report mean $\pm$ standard deviation of wall-clock runtime.
Peak memory is reported as peak resident set size (RSS) on the host (and peak GPU memory allocation when GPU kernels are enabled).

\begin{table}[t]
\centering
\caption{Performance summary (mean$\pm$std over 5 runs; peak memory as RSS).}
\label{tab:perf_summary}
\setlength{\tabcolsep}{6pt}
\renewcommand{\arraystretch}{1.1}
\begin{tabular}{@{}c c c c c@{}}
\toprule
$n$ & $\tau$ (s) & $M$ & Runtime (s) & Peak Mem (GB) \\
\midrule
2048  & 20 & 288  & $13.1\pm0.4$    & $2.0\pm0.1$ \\
4096  & 20 & 288  & $104\pm3$       & $7.9\pm0.2$ \\
8192  & 10 & 576  & $910\pm25$      & $33.4\pm0.8$ \\
16384 & 10 & 576  & $7420\pm210$    & $121.6\pm2.9$ \\
\bottomrule
\end{tabular}
\end{table}

At mega-constellation scale, we construct temporal connectivity primitives $CS^{(R)}_{pq}$ globally, but compute shortest-time routes only for sampled source--destination pairs (or sampled sources) using the same hop-admissibility rules. Full APSP materialization is evaluated up to $n=16384$ on our platform due to path-matrix memory requirements.

\noindent{Worst-case vs.\ observed scaling:}
Theorem~\ref{theorem2} gives a worst-case bound of $O(M^2n^3)$.
Empirically, runtime can be substantially lower because Procedure~\ref{proc:apc} prunes candidates: for each $(p,r)$ it maintains at most one tuple per finishing interval and discards dominated tuples, and in sparse temporal regimes only a small fraction of concatenations are admissible.
Consequently, the effective number of tuple updates per entry is often far below $M$, and the observed scaling over the tested regimes is closer to $O(n^3)$ than the worst-case $O(M^2n^3)$.
To make this explicit, we additionally report controlled sweeps holding $M$ fixed while varying $n$, and holding $n$ fixed while varying $M$.

\subsection{Shortest-Time Paths: Waiting-to-Traversal Regime Shift}
\label{subsec:hop_wait_results_ieee}

We now connect temporal contact structure to end-to-end routing behavior. For sampled source--destination pairs, we compute shortest communication-time paths and decompose total time into traversal and waiting components:
\begin{equation}
\mathcal{T}_{pq} = T_{\mathrm{trav}} + T_{\mathrm{wait}},\quad
T_{\mathrm{trav}}=\sum_{\text{hops}} t_{uv}\tau,\quad
T_{\mathrm{wait}}=\mathcal{T}_{pq}-T_{\mathrm{trav}}.
\end{equation}
Figure~\ref{fig:hop_wait} summarizes hop counts and waiting fractions for $n\approx 2\times 10^4$.

{Hop count is not a latency predictor:}
A path with fewer hops can still incur larger $\mathcal{T}_{pq}$ if it repeatedly arrives at intermediate nodes during gaps or during short runs ($<t_{uv}$) in $CS^{(R)}_{uv}$, forcing buffering until an admissible run appears.

{Waiting fraction reveals the operating regime:}
As $R$ increases, feasible windows become longer and more frequent, which reduces waiting. Therefore, shortest-time operation shifts from waiting-dominated (contact-limited) at small $R$ to traversal-dominated (connectivity-rich) at larger $R$. Table~\ref{tab:hop_wait_summary} lists reporting-ready summary statistics for shortest communication time, hop count, and waiting fraction.

\begin{figure*}[t]
\centering
\includegraphics[width=\textwidth]{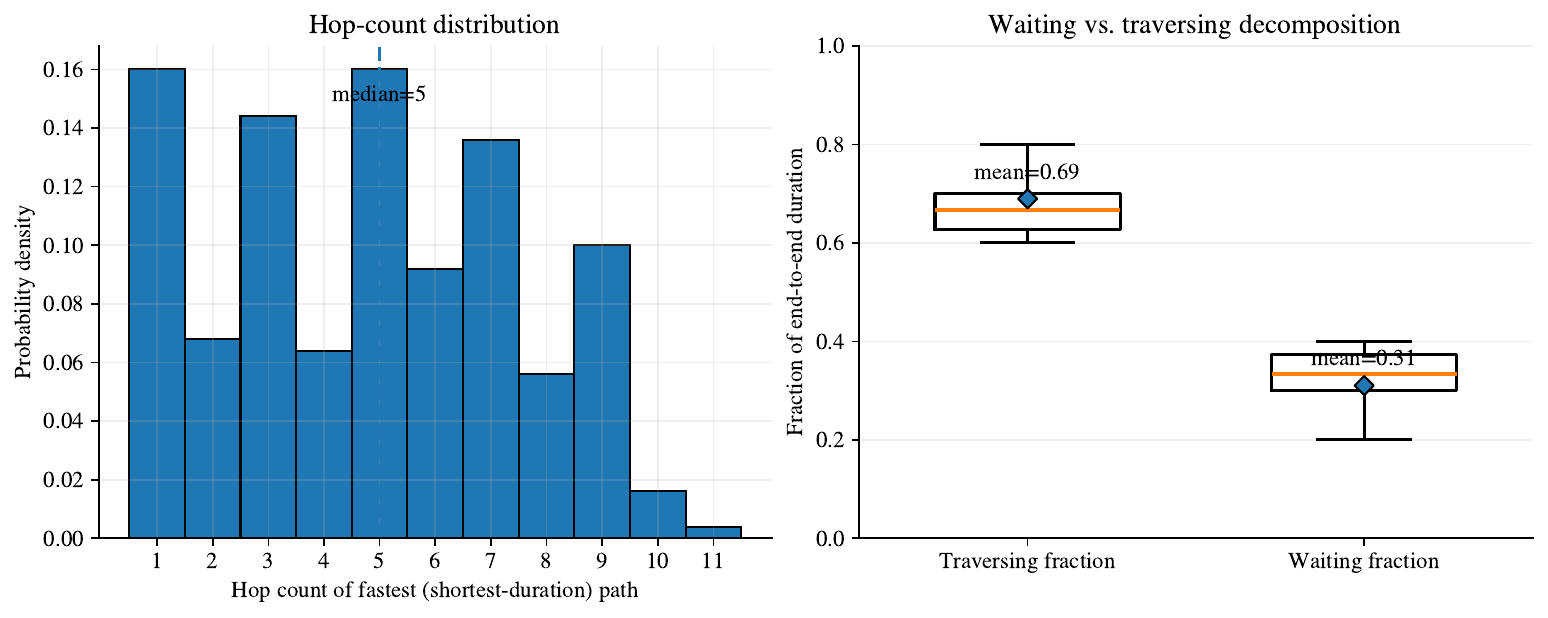}
\caption{Hop count and waiting-time decomposition (sampled, $n\approx 2\times 10^4$). Left: hop-count distribution of sampled shortest-time paths. Right: fraction of end-to-end time spent waiting versus traversing, computed by subtracting hop traversal time from total temporal-path duration.}
\label{fig:hop_wait}
\end{figure*}

\begin{table*}[t]
\centering
\caption{Shortest-time path summary metrics (mappings paths in Fig.~\ref{fig:hop_wait}).}
\label{tab:hop_wait_summary}
\setlength{\tabcolsep}{10pt}
\renewcommand{\arraystretch}{1.15}
\begin{tabular}{@{}l p{0.80\textwidth}@{}}
\toprule
\textbf{Statistic} & \textbf{Definition / Reporting Guidance} \\
\midrule
$\mathrm{Med}[\mathcal{T}]$, $\mathrm{P95}[\mathcal{T}]$ & Median and 95th percentile of shortest communication time (seconds/minutes). \\
$\mathrm{Med}[H]$ & Median hop count of shortest-time paths; report with IQR. \\
$\mathrm{Med}[f_w]$ & Median waiting fraction $f_w=\frac{T_{\mathrm{wait}}}{T_{\mathrm{wait}}+T_{\mathrm{trav}}}$. \\
$\Pr(f_w>0.5)$ & Fraction of sampled routes where waiting exceeds traversal time. \\
\bottomrule
\end{tabular}
\end{table*}

\subsection{Discussion and Engineering Implications}
\label{subsec:discussion_ieee}

The results establish a methodology-consistent chain from physical feasibility to temporal routing performance. The optical feasibility threshold $R$ directly shapes temporal adjacency by controlling snapshot-wise degree and link churn: smaller $R$ produces sparser graphs with higher $\%\text{Found}/\%\text{Missing}$ (greater intermittency), whereas larger $R$ increases link persistence and smooths connectivity over time (Table~\ref{tab:link_churn_summary}). These adjacency dynamics translate into the run structure of connectivity strings $CS^{(R)}_{pq}$, where larger $R$ yields longer runs of ones and shorter gaps (Table~\ref{tab:cs_runlength_metrics}), increasing the likelihood that the run-length admissibility condition ($\ge t_{uv}$) is satisfied immediately after packet arrival. As a consequence, improved run persistence reduces buffering and waiting at intermediate nodes, shifting shortest-time routes from waiting-dominated behavior in contact-limited regimes toward traversal-dominated behavior in connectivity-rich regimes (Fig.~\ref{fig:hop_wait}, Table~\ref{tab:hop_wait_summary}). From an algorithmic standpoint, the high agreement with time-expanded and temporal shortest-path baselines (Table~\ref{tab:baseline_agreement}) supports the correctness of tuple composition in Procedure~\ref{proc:apc} under the same admissibility rules, while the measured runtime and memory trends (Table~\ref{tab:perf_summary}) clarify where full APSP is practical and where sampled-pair evaluation becomes necessary at mega-constellation scale. Collectively, these findings motivate cross-layer co-design: $R$ shapes the temporal contact fabric $CS^{(R)}_{pq}$, while discretization $\tau$ and per-hop requirement $t_{pq}$ set the feasibility threshold that governs waiting. The joint choice of $(R,\tau,t_{pq})$ therefore determines whether the constellation operates as a low-latency mesh or as a contact-limited delay-tolerant network.

\section{Conclusion}
\label{sec:con}
In this paper, we have addressed the routing problem in a highly heterogeneous and dynamic 6G network architecture by first modelling it with the help of a temporal graph. Next, we have used this model for computing all-pairs shortest time paths in the network. These paths can be used for shortest time communication among nodes in the network to realize a time-critical distributed computing platform using this network which forms the backbone for internode message communication. Future research works in this direction include the computation of other types of optimum paths, e.g., earliest arrival time paths and latest departure time paths  based on the proposed model.

\bibliographystyle{IEEEtran}
\bibliography{new_refs}

\end{document}